\documentclass{article}
\usepackage{amsfonts,amssymb,amsthm,dsfont}
\usepackage{amsmath,amssymb,xspace}
\usepackage{fullpage}
\usepackage{graphicx}
\usepackage{subfigure}
\usepackage{hyperref}

\newcommand{\ball}[2]{\ensuremath{\mathcal{B}_{#1}({#2})}}

\newcommand{\qV}{\ensuremath{V}}

\newcommand{\qD}{\ensuremath{\Delta}}

\newcommand{\qL}{\ensuremath{L}}

\newcommand{\qB}{\ensuremath{B}}

\newcommand{\qR}{\ensuremath{R}}

\newcommand{\qQal}{\ensuremath{\mathcal{Q}}}

\newcommand{\qS}{\ensuremath{S}}

\newcommand{\qTal}{\ensuremath{\mathcal{T}}}

\newcommand{\qQop}{\ensuremath{Q}}

\newcommand{\qq}{\ensuremath{q}}

\newcommand{\qthr}{\ensuremath{t}}
\newcommand{\qlev}{\ensuremath{l}}

\newcommand{\qcell}{\ensuremath{\mathcal{C}}}

\newcommand{\qerr}{\ensuremath{\varepsilon}}

\newcommand{\pdf}{\ensuremath{\varphi}}

\newcommand{\ie}{\emph{i.e.}, }
\newcommand{\eg}{\emph{e.g.}, }
\newcommand{\inv}[1]{\frac{1}{#1}}
\newcommand{\tinv}[1]{\tfrac{1}{#1}}
\newcommand{\ud}{{\rm d}} % measure integration "d"

\newcommand{\bb}{\mathbb}

\newcommand{\cl}{\mathcal}

\newcommand{\st}{\ {\rm s.t.}\ }

\newcommand{\Id}{I}

\newcommand{\noise}{\xi}

\newcommand{\fnorm}[1]{|\!|\!|#1|\!|\!|}

\DeclareMathOperator{\diag}{diag}

\DeclareMathOperator{\sign}{sign}

\newcommand{\sd}{\ensuremath{\Sigma\Delta}\xspace}

\newcommand{\decoder}{\ensuremath{\mathcal D}}

\newtheorem{theorem}{Theorem}
\newtheorem{definition}{Definition}
\newtheorem{proposition}{Proposition}

\begin{document}

\title{Quantization and Compressive Sensing}
\author{Petros T. Boufounos\footnote{Mitsubishi Electric Research
  Laboratories, 201 Broadway, Cambridge, MA, USA, \url{petrosb@merl.com}}, Laurent Jacques\footnote{ ISPGroup, ICTEAM/ELEN, Universit\'{e} catholique de Louvain, Place 
du Levant 2, PO box L5.04.04, B1348 Louvain-la-Neuve, Belgium \url{laurent.jacques@uclouvain.be}}, Felix Krahmer\footnote{Georg-August-Universit\"at G\"ottingen, Lotzestra{\ss}e 16-18, 37083 G\"ottingen, Germany \url{f.krahmer@math.uni-goettingen.de}}, and Rayan Saab\footnote{Univeristy of California, San Diego, 9500 Gilman
Drive \#0112, La Jolla, CA 92093-0112, USA \url{rsaab@ucsd.edu}}}

\maketitle

\begin{abstract}
  Quantization is an essential step in digitizing signals, and,
  therefore, an indispensable component of any modern acquisition
  system. This chapter explores the interaction of quantization and
  compressive sensing and examines practical quantization strategies
  for compressive acquisition systems. Specifically, we first provide
  a brief overview of quantization and examine fundamental performance
  bounds applicable to any quantization approach. Next, we consider
  several forms of scalar quantizers, namely uniform, non-uniform, and
  1-bit. We provide performance bounds and fundamental analysis, as
  well as practical quantizer designs and reconstruction algorithms
  that account for quantization. Furthermore, we provide an overview
  of Sigma-Delta (\sd) quantization in the compressed sensing context,
  and also discuss implementation issues, recovery algorithms and
  performance bounds. As we demonstrate, proper accounting for
  quantization and careful quantizer design has significant impact in
  the performance of a compressive acquisition system.
\end{abstract}

\section{Introduction}
\label{sec:intro}

In order to store and manipulate signals using modern devices, it is
necessary to digitize them. This involves two steps: sampling (or
measurement) and quantization. The compressed sensing theory and
practice described in the remainder of this book provides a novel
understanding of the measurement process, enabling new technology and
approaches to reduce the sampling burden. This chapter explores the
very interesting interaction of compressed sensing with quantization.

Sampling maps a signal to a set of coefficients, typically using
linear measurements. This map can often be designed to be lossless,
\ie to perfectly represent all signals in a certain class, as well as
robust to noise and signal modeling errors. The Nyquist
theorem, as well as more recent compressive sampling theorems are
examples of such sampling
approaches~\cite{candes2006ssr,blu2008sparse,mishali2011sub}.

The guarantees in sampling theorems are typically stated in terms of
the critical measurement rate, \ie the number of measurements necessary to
perfectly represent signals in a given class. Oversampling, compared
to that minimum rate, typically provides robustness to errors in the
representation, noise in the acquisition, and mismatches in signal
models. The latter is especially important in compressive sensing
systems as they provide perfect reconstruction guarantees for exactly
sparse signals; in practice, the acquired signal is almost never
exactly sparse.

Quantization, on the other hand, is the process of mapping the
representation coefficients---which potentially belong to an
uncountably infinite set---to elements in a finite set, and
representing them using a finite number of bits. Due to the
many-to-one nature of such a map, the quantized representation is in
general lossy, \ie distorts the representation and, therefore, the
signal. This distortion occurs even if the measurement process is lossless.

The interaction of quantization with sampling introduces interesting
trade-offs in the acquisition process. A system designed to sample 
signals at (or slightly above) the critical rate may be less
robust to errors introduced by quantization. Consequently, it requires a
sophisticated quantizer design that ensures very small quantization errors. 
On the other hand, a simpler quantizer architecture (e.g., with fewer bits per measurement) could introduce significant error to the representation
and require some oversampling to compensate. Practical systems
designs navigate this trade-off, for example, according to the complexity of the
corresponding hardware.

Compressive acquisition systems amplify the importance of the
trade-off between quantizer complexity and oversampling. The sampling rate is significantly reduced in such systems, at the expense of increased sensitivity to noise and signal model
mismatch. Thus, loss of information due to quantization
can be detrimental, especially when not properly handled. One may revert to oversampling here as well, however
the incoherent and often randomized nature of compressive
measurements poses challenges.  Thus, powerful oversampling based quantization approaches, such as Sigma-Delta quantization can be applied, but only after careful consideration.

Nevertheless, the sparse signal models and the computational
methods developed for compressed sensing can alleviate a number of
performance bottlenecks due to quantization in conventional
systems. Using computational approaches originating in frame theory
and oversampling, it is possible to significantly reduce the
distortion due to quantization, to significantly improve the
performance due to saturation, and to enable reconstruction from
measurements quantized as coarsely as 1~bit. The theory and practice
for such methods are described in Sec.~\ref{sec:scalarCS}.

It might seem counter-intuitive that compressed sensing
attempts to remove sampling redundancy, yet successful reconstruction
approaches employ tools developed for oversampled representations. In
fact there is a strong connection between compressed sensing and
oversampling, which we explore in various points in this
chapter. Furthermore, with sufficient care, this connection can be
exposed and exploited to implement Sigma-Delta quantization in
CS-based acquisition systems, and significantly improve performance
over scalar quantization. The details are discussed in
Sec.~\ref{sec:SDCS}.

The next section presents general principles of quantization,
including a brief background on vector, scalar, and Sigma-Delta
quantization for general acquisition systems. It is not an exhaustive
survey of the topic. For this we refer the reader
to~\cite{gray1998quantization,DD03,NST96}. Instead, it serves to
establish notation and as quick reference for the subsequent
discussion. Sec.~\ref{sec:scalarCS} and Sec.~\ref{sec:SDCS} examine the
interaction of compressive sensing and quantization in significant
detail. Sec.~\ref{sec:discussion} concludes with some discussion of
the literature, promising directions and open problems.

\ \\
\noindent\textbf{Notation:} In addition to the notational conventions
defined in Chapter 1, this chapter also uses the following general
notations. The logarithm in base $a>0$ is noted $\log_a$ and whenever
the base is not specified, $\log$ refers to the natural
logarithm. Note that in some cases, such as asymptotic results, the
logarithm base is not important. This chapter also uses the following
non-asymptotic orderings: For two functions $f$ and $g$, we write $f \lesssim
g$ if there exists a constant $C>0$ independent of the function
arguments such that $f \leq C g$, with a similar definition for $f
\gtrsim g$. Moreover, $f \asymp g$ if we have both $f \lesssim g$ and
$f \gtrsim g$.  Occasionally, we also rely on the well-established
big-$O$ and big-$\Omega$ asymptotic notation to concisely explain
asymptotic behavior when necessary. More specific notation is defined
at first occurrence.

\section{Fundamentals of Quantization}
\label{sec:intro_fundamentals}
For the purposes of this section, a quantizer operates on signals $x$, viewed as vectors
in a bounded set $\qV\subset\mathbb{R}^n$. The goal of a quantizer
$\qQop(\cdot)$ is to represent those signals as accurately as possible
using a rate of \qR\ bits, \ie using a quantization point
$\qq=\qQop(x)$ chosen from a set of $2^\qR$ possible ones often
referred to as codebook. Of course, when $\qV$ contains an infinite
number of signals, signals will be distorted through this
representation.

In this section, we first define common quantization performance
metrics and determine fundamental bounds on the performance of a
quantizer. Then, in preparation for the next sections, we examine
common approaches to quantization, namely scalar and Sigma-Delta
quantization, which are very useful in compressive sensing
applications.

\subsection{Quantization Performance Bounds}
\label{sec:intro_vector}
To measure the accuracy of the quantizer we consider the distortion, i.e.,  the $\ell_2$ distance of a
quantization point from its original signal $\|x-\qQop(x)\|_2$. The
overall performance of the quantizer is typically evaluated either using the
average distortion over all the signals---often computed using a
probability measure on the signal space \qV---or using the worst case
distortion over all signals in \qV. In this chapter, in the spirit of
most of the compressed sensing literature, we quantify the performance
of the quantizer using the worst case distortion on any signal, \ie
\begin{align}
\qerr=\sup_{x\in \qV}\|x-\qQop(x)\|_2.
\end{align}
This choice enables very strong guarantees, irrespective of the
accuracy of any probabilistic assumption on the signal space.

A lower bound on the distortion of any quantizer can be derived by
constructing a covering of the set \qV. A covering of radius $r$ is a
set of points \qq\ such that each element in \qV\ has distance at most
$r$ from its closest point in the covering. If we can construct a
covering using $P$ points, then we can also define a quantizer that
uses $\qR=\lceil\log_2 P\rceil$ bits and has worst case distortion $\qerr=r$ as
each signal is quantized to the closest point in the covering.

To determine a lower bound for the number of points in such a
covering, we consider balls of radius $r$ centered {at} $\qq$, defined as
\begin{align}
  \ball{r}{\qq}=\left\{\left.x\in\mathbb{R}^n\right|\|\qq-x\|_2\le r\right\}.
\end{align}
Since each signal in \qV\ is at most $r$ away from some point in the
covering, if we place a ball of radius $r$ at the center of each point
of the covering, then the union of those balls covers \qV. Thus, the total volume of the balls should be at
least as large as the volume of the set, denoted
$\mathrm{vol}(\qV)$. Since the volume of a ball of radius $r$ in $n$
dimensions is
$\mathrm{vol}(\ball{r}{\qq})=r^n\pi^{n/2}/\Gamma(1+n/2)$, where
$\Gamma(\cdot)$ is the Gamma function, the best possible error given
the rate \qR\ can be derived using
\begin{align}
  \mathrm{vol}(\qV)\le
  \tfrac{\pi^{n/2}}{\Gamma\left(1+\frac{n}{2}\right)}\,2^\qR\,r^n
  \Rightarrow r \gtrsim 2^{-\frac{\qR}{n}}.
\end{align}
In other words, the worst-case error associated with an optimal quantizer can, at best, 
decay exponentially as the bit rate increases. Moreover, the decay rate
depending on the ambient dimension of the signal. In short,
\begin{align} \qerr&\gtrsim 2^{-\frac{\qR}{n}}.
\end{align}
The smallest achievable worst case distortion for a set is also known
as the $(\qR+1)$-dyadic entropy number of the set, whereas the number
of bits necessary to achieve a covering with worst-case distortion
equal to $\qerr$ is known as the Kolmogorov $\qerr$-entropy or metric
entropy of the set.

For the models commonly assumed in compressive sensing, these
quantities are not straightforward to calculate and depend on the
sparsity model assumed. For example, compressible signals are commonly modeled as 
being drawn from a
unit $\ell_p$ ball, where $0<p<1$ (cf.~Chapter 1 for a discussion on compressibility). In this case, the worst case distortion is bounded by
\begin{align}
  \qerr\ \gtrsim\ \left\{
  \begin{array}{cl}
    1&\mbox{if}~1\le \qR\le \log_2 n\\    \,\left(\tinv{\qR}\log_2(\tfrac{n}{\qR}+1)\right)^{\inv{p}-\inv{2}}&\mbox{if}~\log_2 n\le \qR\le n\\[1mm]
    \,2^{-\frac{\qR}{n}}n^{\inv{2}-\inv{p}}&\mbox{if}~\qR\ge n,
  \end{array}
  \right.  
\end{align}
where the constant implicit in our nation is independent of \qR\ and
$n$~\cite{schutt1984entropy,edmunds1996function,kuhn2001lower,candes2006encoding}.

In the case of exactly $k$-sparse signals, the volume of the union of
subspaces they occupy has measure zero in the $n$-dimensional ambient
space. However, by considering the $\binom{n}{k}$ $k$-dimensional
subspaces and coverings of their unit balls, a lower bound on the error can be
derived~\cite{BB_DCC07}, namely
\begin{align}
  \qerr\gtrsim \frac{2^{-\frac{\qR}{k}}n}{k}.
  \label{eq:sparse_lower}
\end{align}
Note that this lower bound can be achieved in principle using standard
transform coding (TC), \ie by first representing the signal using its
sparsity basis, using $\log_2\binom{n}{k} \lesssim k\log_2(n/k)$ bits to
represent the support of the non-zero coefficients and using the
remaining bits to represent the signal in the $k$-dimensional subspace
at its Kolmogorov entropy
\begin{align}
  \qerr_{\mathrm{TC}} \lesssim 2^{-\frac{\qR-k\log_2(n/k)}{k}} =\frac{2^{-\frac{\qR}{k}}n}{k}.
\end{align}
Unfortunately, compressive sensing systems do not have direct access to
the sparse vectors. They can only access the measurements, $y=Ax$,
which must be quantized upon acquisition---in practice using analog
circuitry. Thus, transform coding is not possible. Instead, we must
devise simple quantization algorithms that act directly on the
measurements in such a way that permits accurate reconstruction.
\subsection{Scalar Quantization}
\label{sec:intro_scalar}
The simplest approach to quantization is known as {\em scalar
  quantization} and often referred to as {\em pulse code modulation}
(PCM), or {\em memoryless scalar quantization} (MSQ). Scalar
quantization directly quantizes each measurement of the signal,
without taking other measurements into account. In other words a
1-dimensional, \ie scalar, quantizer is applied separately to each
measurement of the signal.

\subsubsection{Measurement and Scalar Quantization}
\label{sec:meas-scal-quant}

A scalar quantizer can be defined using a set of levels, $\qQal
=\{\qlev_i \in \bb R: \qlev_j < \qlev_{j+1}\}$, comprising the
quantization codebook, and a set of thresholds $\qTal = \{\qthr_i \in
\overline{\bb R}: \qthr_j < \qthr_{j+1}\}$, implicitly defining the
quantization intervals $\qcell_j = [\qthr_j,\qthr_{j+1})$. Assuming no
  measurement noise, the quantizer is applied element-wise to the
  measurement coefficients, $y=Ax$, to produce the quantized
  measurements $\qq=Q(y),~\qq_i=Q(y_i)$. Using a rate of \qB\ bits per
  coefficient, \ie $\qR=m\qB$ total bits, the quantizer represents
  $\qL=2^\qB$ total levels per coefficient. A scalar value $y_i$
  quantizes to the quantization level corresponding to the
  quantization interval in which the coefficient lies.
\begin{equation}
  \label{eq:scalar-quant-def}
  \qQop(y_i) = \qlev_j\quad\Leftrightarrow\quad y_i \in \qcell_j.  
\end{equation}

A scalar quantizer is designed by specifying the quantization levels
and the corresponding thresholds. Given a source signal with
measurements modeled as a continuous random variable $X$, a
\emph{(distortion) optimal} scalar quantizer minimizes the error
\begin{equation}
  \label{eq:distortion-optimal}
  \bb E|X - \qQop(X)|^2. 
\end{equation}
Such an optimal quantizer necessarily satisfies the Lloyd-Max
conditions~\cite{Lloyd1982,Max1960}
\begin{equation}
  \label{eq:LM-condition}
  \qlev_j = \bb E\left\{X|X \in \qcell_j\right\},\qquad  \qthr_j =
  \tinv{2}(\qlev_j+\qlev_{j+1}),
\end{equation}
which define a fixed point equation for levels and thresholds and the
corresponding fixed-point iteration---known as the Lloyd-Max
algorithm---to compute them.

\begin{figure}[t]
  \centerline{\includegraphics[width=.7\linewidth]{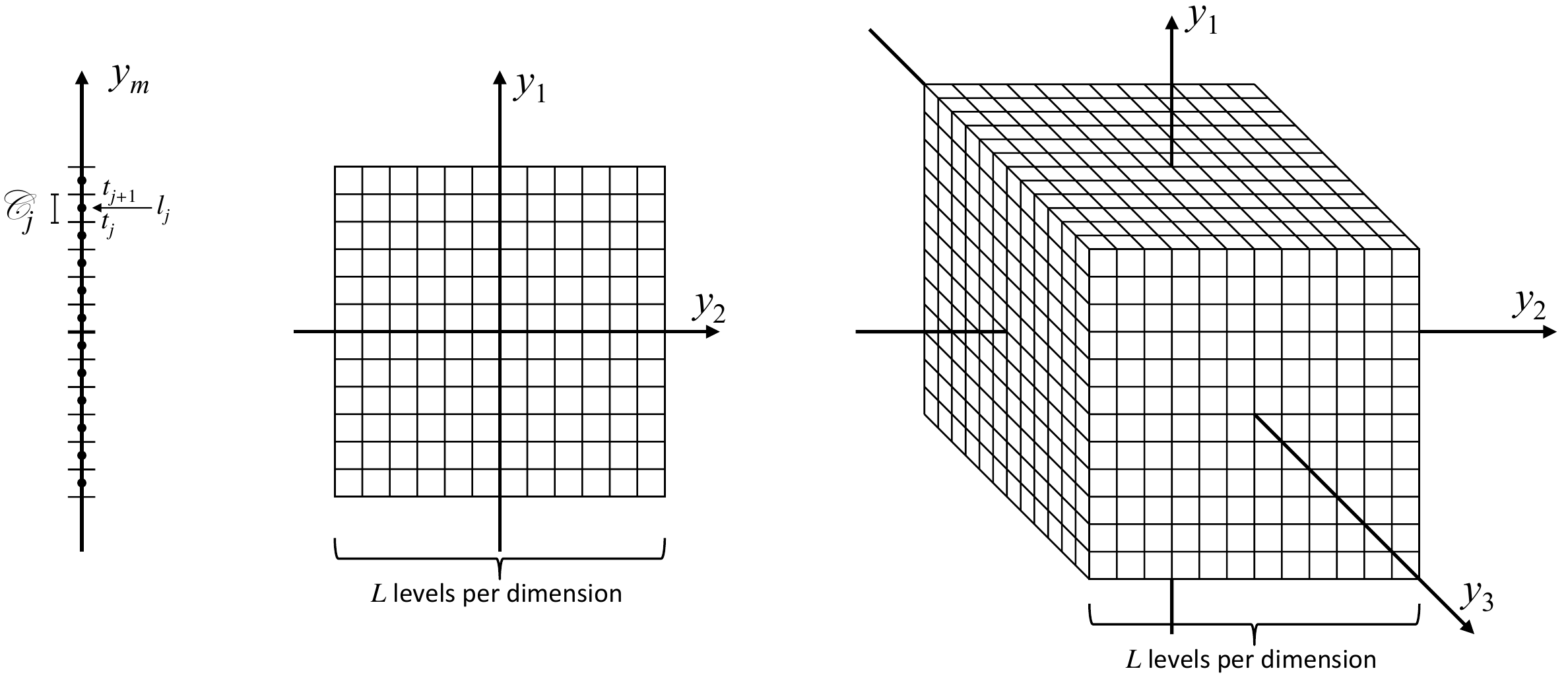}}
  \caption{A finite uniform scalar quantizer and the uniform grid it
    generates in 2 and 3 dimensions}
  \label{fig:quant_grid}
\end{figure}

Alternatively, a simpler design approach is the uniform scalar
quantizer, which often performs almost as well as an optimal scalar
quantizer design. It is significantly less complex and can be shown to
approach optimality as the bit-rate increases~\cite{gray1998quantization}. The
thresholds of a uniform scalar quantizer are defined to be
equi-spaced, \ie $\qthr_{j+1} - \qthr_j = \qD$, where $\qD$ is
referred to as the quantization bin width or \emph{resolution}. The
levels are typically set to the mid-point $\qlev_j = \tinv{2}(\qthr_j
+ \qthr_{j+1})$ of the quantization bin $\qcell_j$. Thus, the
quantization error introduced to each coefficient is bounded by
$\qD/2$. A uniform quantizer defines a uniform grid in the
$m$-dimensional measurement space, as shown in
Fig.~\ref{fig:quant_grid}.

In practical systems, the scalar quantizer has finite range, \ie it
saturates if the signal exceeds a saturation level \qS. In particular,
a uniform finite-range scalar quantizer using $B$ bits per coefficient
has quantization interval $\qD=\qS 2^{-B+1}$. If a coefficient exceeds
$S$, the quantizer maps the coefficient to the largest quantization
level, \ie it saturates.  Depending on the magnitude of the
coefficient, this may introduce significant error. However, it is
often convenient in theoretical analysis to assume an infinite
quantizer that does not saturate. This assumption is often justified,
as \qS\ in practice is set large enough to avoid saturation given a signal
class. As described in Sec.~\ref{sec:unif-scal-quant-with-saturation},
this is often suboptimal in compressive sensing applications.

Compared to classical systems, optimal scalar quantizer designs for
compressive sensing measurements require extra care. An optimal design with respect to the measurement error is
not necessarily optimal for the signal, due to the non-linear reconstruction inherent in compressed sensing. While specific designs have been derived for very
specific reconstruction and probabilistic signal models,
\eg~\cite{sun2009optimal,kamilov2011optimal}, a general optimal design
remains an open problem. Thus the literature has focused mostly, but
not exclusively, on uniform scalar quantizers. 

\begin{figure}[t]
  \centerline{\includegraphics[width=.8\linewidth]{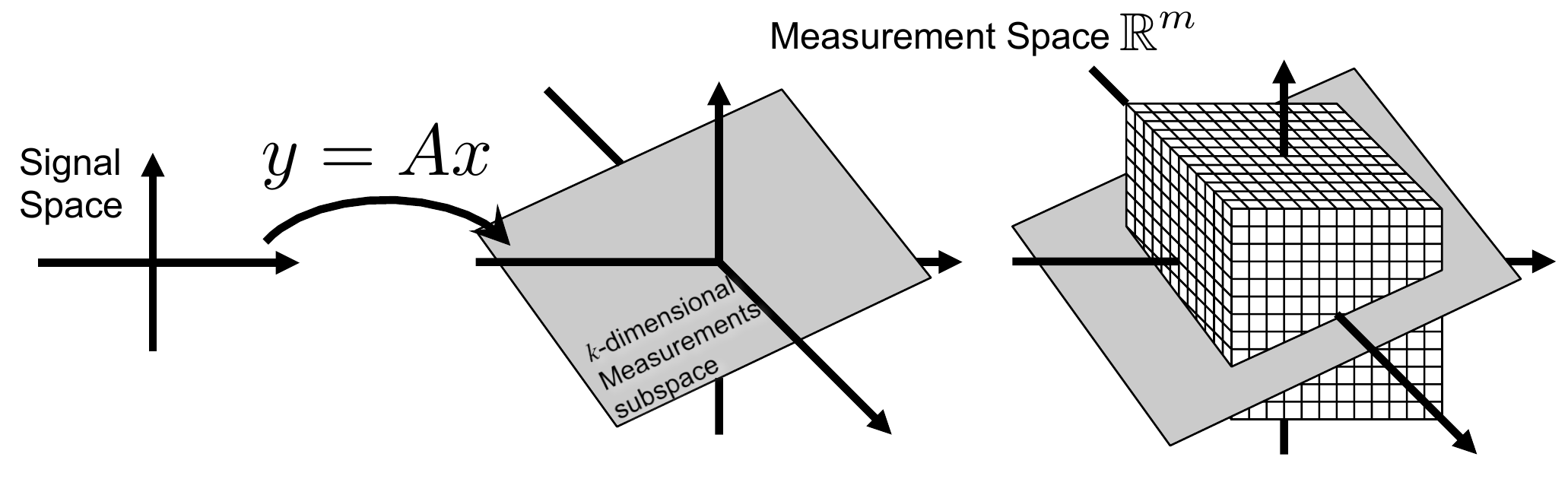}}
  \caption{A $k$-dimensional space measured using $m$ measurements
    spans a $k$-dimensional subspace of $\mathbb{R}^m$ and intersects
    only a few of the $L^m$ available quantization cells.}
  \label{fig:cell_intersect}
\end{figure}

\subsubsection{Scalar Quantization and Oversampling}
\label{sec:scal-quant-overs}

When a signal is oversampled, a scalar quantizer makes suboptimal use
of the bit-rate. The $k$-dimensional signal space mapped through the
measurement operator to an $m$-dimensional measurement space, where
$m>k$, spans, at most, a $k$-dimensional subspace of $\mathbb{R}^m$,
as shown in Fig.~\ref{fig:cell_intersect}. As evident from the figure,
this subspace intersects only a few of the available quantization
cells and, therefore, does not use the available bits effectively. For
an $\qL$-level quantizer, the number of quantization cells intersected
$I_{k,m,L}$ is bounded
by~\cite{bib:Thao96,goyal_1998_lowerbound_qc,bib:Boufounos06thesis}
\begin{align}
  I_{k,m,L}\lesssim\left(\frac{Lm}{k}\right)^k
\end{align}
Using a simple covering argument as in Sec.~\ref{sec:intro_vector}, it
is thus possible to derive a lower bound on the error performance as a
function of the number of measurements $m$
\begin{align}
  \qerr\gtrsim\left(\frac{2^{-\qB}k}{m}\right) \label{eq:scalar_bound1}
\end{align}
The bounds hold for any scalar quantizer design, not just uniform
ones.

Linear reconstruction, \ie reconstruction using a linear operator
acting on the scalar quantized measurements, does not achieve the bound
~\eqref{eq:scalar_bound1}~\cite{bib:Thao94,goyal_1998_lowerbound_qc}. The
quantization error using linear reconstruction can only decay as fast
as 
\begin{equation}
  \qerr\gtrsim\frac{2^{-\qB}k}{\sqrt{m}}.
  \label{eq:scalar_bound2}
\end{equation} 

Instead, {\em consistent}
reconstruction achieves the optimal bound in a number of
cases. Consistent reconstruction treats the quantization regions as reconstruction
constraints and ensures that the reconstructed signal $\hat{x}$
quantizes to the same quantization points when measured using the same
system. Thus in the oversampled setting where $A$ is an $m\times k$
matrix with $m>k$, and where $\qq=\qQop(Ax)$ one solves the problem:
\begin{align}
\label{eq:consistent-algo}
  \mathrm{find~any}~~\hat{x}&~~\mathrm{s.t.}~~\qq=\qQop(A\hat{x}).
\end{align}
If the measurement operator $A$ is a tight frame formed by an
  oversampled Discrete Fourier Transform (DFT), the root mean square
  error (RMSE) of such a reconstruction (with respect to a random
  signal model) decays as $O(1/m)$
  \cite{bib:Thao94,goyal_1998_lowerbound_qc}, \ie as
  \eqref{eq:scalar_bound1}. In the case of random
  frames with frame vectors drawn independently from a Gaussian distribution \cite{jacques2014error} or from a suitable distribution on the
  $(m-1)$-sphere \cite{powell_consistent}, the reconstruction method
  in \eqref{eq:consistent-algo}
  also displays RMSE and worst
  case reconstruction error decreasing as $O(1/m)$ and $O((\log m)/m)$, respectively.

The constraints imposed by consistent reconstruction are convex and
can be imposed on any convex optimization algorithm. This makes them
particularly suitable for a number of reconstruction algorithms
already used in compressive sensing systems, as we explore in
Sec.~\ref{sec:scalarCS}.

The bounds \eqref{eq:scalar_bound1} and
\eqref{eq:scalar_bound2}---which can be achieved with proper design of
the measurement process and the reconstruction algorithm---demonstrate
that the most efficient use of the rate $\qR=m\qB$ is in refining each
measurement using more bits per measurement, \qB, rather than in
increasing the number of measurements, $m$. They suggest that in terms
of error performance, by doubling the oversampling it is possible to
save 0.5 bits per coefficient if linear reconstruction is used and 1
bit per coefficient if consistent reconstruction is used. This means
that a doubling of the rate by doubling the oversampling factor, is
equivalent to a linear increase in the rate by $m/2$ or $m$ through an
increase in \qB, for linear and consistent reconstruction,
respectively. So in principle, if rate-efficiency is the objective, the
acquisition system should only use a sufficient number of measurements
to reconstruct the signal and no more. All the rate should be devoted
to refining the quantizer. However, these bounds ignore the practical advantages in oversampling
a signal, such as robustness to erasures, robustness to measurement
noise and implementation complexity of high-rate scalar quantizers. Thus in
practice, oversampling is often preferred, despite the
rate-inefficiency. Techniques such as Sigma-Delta quantization, which we discuss in Sec.~\ref{sec:intro_SD}, have been developed to
improve some of the trade-offs and are often used in conjunction with oversampling.

\subsubsection{Performance Bounds on Sparse Signals}
\label{sec:perf-bounds-sparse}

Scalar quantization in compressive sensing exhibits similar bounds as
scalar quantization of oversampled signals. Signals that are
$k$-sparse in $\bb R^n$ belong
to a union of $k$-dimensional subspaces. When measured using $m$
linear measurements, they occupy a union of $k$-dimensional subspaces
of $\mathbb{R}^m$, $\binom{n}{k}$ of them. Using the same counting
argument as above, it is evident that the number of quantization cells
intersected, out of the $L^m$ possible ones, is at most
\begin{align}
  \binom{n}{k}I_{k,m,L}\gtrsim\left(\frac{Lmn}{k^2}\right)^k
\end{align}
The resulting error bound is
\begin{align}
  \qerr&\gtrsim\frac{2^{-\qB}k}{m}\\
  &\gtrsim\frac{2^{-\frac{\qR}{m}}k}{m},
\end{align}
which decays slower than \eqref{eq:sparse_lower} as the rate increases
keeping the number of measurements $m$ constant. Furthermore, as the
rate increases with the number measurements $m$, keeping $\qB$, the
number of bits per measurement constant, the behavior is similar to
quantization of oversampled frames: the error can only decay linearly
with $m$.

These bounds are not surprising, considering the similarities of
oversampling and compressive sensing of sparse signals. It should,
therefore, be expected that more sophisticated techniques, such as
Sigma-Delta ($\Sigma\Delta$) quantization should improve performance,
as they do in oversampled frames. However, their application is not as
straightforward. The next section provides an overview of
$\Sigma\Delta$ quantization and Sec.~\ref{sec:SDCS} discusses in
detail how it can be applied to compressive sensing.
\subsection{Sigma-Delta Quantization}
\label{sec:intro_SD}
An alternative approach to the scalar quantization techniques detailed
in the previous section is feedback quantization. The underlying idea
is that the fundamental limits for the reconstruction accuracy
discussed above can be overcome if each quantization step takes into
account errors made in previous steps. The most common feedback
quantization scheme is $\sd$ quantization, originally introduced for
bandlimited signals in \cite{IYM62} (cf. \cite{IY63}). A simple $\sd$ scheme, illustrated in Figure \ref{fig:SD_blk}, shows this feedback structure.
\begin{figure}
\begin{center}
\includegraphics{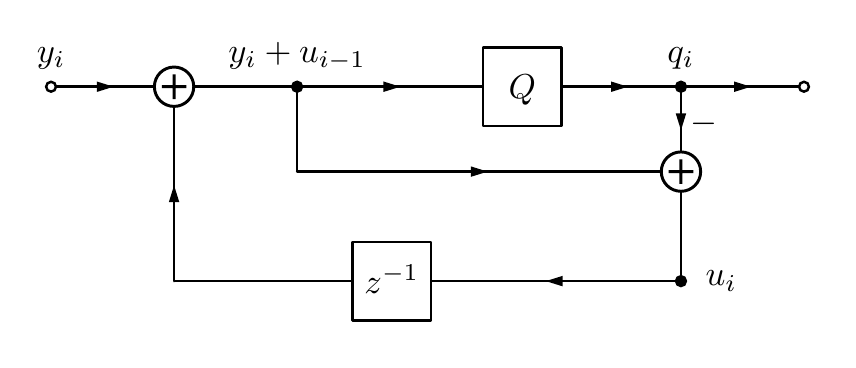}\label{fig:SD_blk}
\caption{A block diagram of a simple $1$st order $\sd$ scheme: The input $y_i$ is added to the state variable $u_{i-1}$ (initialized as $u_0 =0 $) and the sum is scalar quantized. Subsequently, the state variable is updated as the difference between the scalar quantizer's input and its output. More complex designs, featuring higher order $\sd$ quantization with more feedback loops are possible. We discuss such designs in more detail in  Section~\ref{sec:SDCS}.}
\end{center}
\end{figure}

A motivation
in \sd\ quantization is that, in some applications, reducing circuit
complexity is desirable, even at the expense of a higher sampling
rate. Indeed, \sd\ designs drastically reduce the required bit depth
per sample while allowing for accurate signal reconstruction using
simple circuits. In fact, since its introduction, $\sd$ quantization has seen widespread use (see,
\eg \cite{NST96} and the references therein) in applications ranging from audio coding to wireless communication. 

Nevertheless, a mathematical analysis of $\sd$ quantization in its full generality has been challenging. A preliminary
analysis of simple $\sd$ schemes for restricted input
classes (including constant input and sinusoidal input) was presented in \cite{G87} and follow-up works. However, most
of these results were limited to linear, or at best low-order polynomial error
decay in the oversampling rate.  This type of error decay is sub-optimal (albeit better than scalar quantization), and  rather far from the optimal exponential error decay.
Specifically, a major difficulty that prevented a
more comprehensive treatment was understanding the instabilities caused by
the positive feedback inherent to the $\sd$ circuit designs. For example, depending on the design of the $\sd$ scheme, the state variables could grow without bound. A crucial idea to
prevent such phenomena for arbitrary \emph{band-limited} inputs was
developed in \cite{DD03}; their analysis led, for the first-time, to
super-polynomial bounds for the error decay. To date, the best known
error bounds decay exponentially in the oversampling rate \cite{G03,
  DGK11}. While this is near-optimal (optimal up to constants in the exponent), it has been shown that with a fixed
bit budget per sample, the achievable rate-distortion relationship is
strictly worse than for scalar quantization of Nyquist rate samples
\cite{KW12}. That said, increasing the bit budget per sample entails more expensive and complex circuitry, which grows increasingly costly with every added bit (in fact, the best current quantizers provide a resolution of about 20 bits per sample). Thus, for quantizing bandlimited functions, if one wishes to improve the performance or reduce the cost, one must revert to oversampling-based methods such as $\sd$ quantization.

The accuracy gain of $\sd$ quantization is most prominent when a
significant oversampling rate and, therefore, a high redundancy of
samples is inherent or desired. Such redundant representations can
also be encountered in a finite-dimensional discrete context. Namely,
this corresponds to a finite frame expansion in the sense of
(1.32). This observation served as a motivation to devise $\sd$
schemes for finite-frame expansions, and the first such construction was
provided in \cite{BPY06}.  In contrast to oversampled representations
of bandlimited signals, which directly correspond to a temporal
ordering, finite frames generally do not have an inherent order, nor
are the frame vectors necessarily close enough to each other to allow
for partial error compensation.  Due to this difficulty, the first
works on $\sd$ quantization for finite frame expansions focus on
frames with special smoothness properties. Namely, they assume that
the frame $\Phi=\{\phi_j\}_{j=1}^N$ has a well controlled {\em frame
  variation} \[v_\Phi:=\sum_{j=1}^{N-1} \|\phi_{j+1}-\phi_j\|_2.\] The
constructions in \cite{BPY06} coupled with (linear) reconstruction via
the canonical dual frame (that is, the Moore-Penrose pseudo-inverse of the matrix that generates the redundant representation) was shown to yield an error decay on the
order of $v_{\Phi} N^{-1}$, \ie linear error decay whenever the frame
variation is bounded by a constant. By using more sophisticated $\sd$
schemes these results were later improved to higher order polynomial
error decay \cite{BPY06-2, BP07, BPA07} in the number of measurements,
thereby beating the bound \eqref{eq:scalar_bound1} associated with
scalar quantization. Again, these constructions require certain
smoothness conditions on the frame and employ the canonical dual frame
for recovery. In a slightly different approach, the design of the feedback 
and the ordering of the frame vectors has been considered as part of the quantizer design~\cite{BO_JASP06,bib:Boufounos06thesis}.

A new take on the frame quantization problem was initiated in
\cite{LPY10, BLPY10} where the authors realized that reconstruction
accuracy can be substantially improved by employing an appropriate
alternative dual frame (\ie, a different left-inverse) for recovery. At the core of this approach is
still a smoothness argument, but this time for the dual frame. Given a
frame, an appropriate dual frame, the so-called Sobolev dual, can be
obtained by solving a least-squares problem over the space of all
duals \cite{ BLPY10}. Again, this yields polynomial error decay,
albeit now in more general settings. Moreover, by optimizing over such
constructions, root-exponential error decay can be achieved
\cite{KSW12}.

While the definition of the Sobolev dual does not require any
smoothness of the frame, the concrete examples discussed in the
aforementioned works still exclusively focused on smooth
frames. Similar results on recovery guarantees for frames without
smoothness properties were first obtained for frames consisting of
independent standard Gaussian vectors \cite{GLPSY13} and subsequently
generalized to vectors with independent subgaussian entries
\cite{KSY13}.

The underlying constructions also form the basis for the $\sd$
quantization schemes for compressed sensing measurements. Details on
such schemes are given in Sec.~\ref{sec:SDCS}. The insight behind
the schemes is that the number of measurements taken in compressed
sensing is typically larger than the support size by at least a
logarithmic factor in the dimension, and there is an interest in
choosing it even larger than that, as this induces additional
stability and robustness. Thus, once the support of the signal has been
identified and only the associated signal coefficients need to be
determined, one is dealing with a redundant representation. The goal
is now to employ frame quantization schemes to exploit this
redundancy.

For typical compressed sensing matrices, any $k$ columns indeed form a
frame; this follows for example from the restricted isometry property.
However, as the support of the signal is not known when quantizing the
measurements, it is crucial that $\sd$ quantization is universal. That
is, it must not require knowledge regarding which of a given
collection of frames (namely, those forming the rows of an $m \times
k$ submatrix of $A$) has been used for encoding. The reconstruction
from the resulting digital encodings then typically proceeds in two
steps. First the support is identified using standard compressed
sensing recovery techniques, just treating the quantization error as
noise. In a second step, only the restriction of the measurement
matrix to the identified support columns is considered. For the frame
consisting of the rows of this matrix, one then applies frame
quantization reconstruction techniques. Recovery guarantees for such
an approach have been proven for Gaussian measurements \cite{GLPSY13}
and measurements with independent subgaussian entries \cite{KSY13}. It
is of great importance that the dual frame used for recovery is chosen
properly (\eg the Sobolev dual), as it follows from the RIP that the
frames never have a small frame variation. Here again the recovery
error bounds decay polynomially in the number of measurements and beat
the analogous bounds for scalar quantization.

Preliminary steps towards a unified approach to support and signal
recovery have been considered in \cite{C13}. The reconstruction
techniques studied in this work, however, intrinsically rely on
certain non-convex optimization problems, for which no efficient
solution methods are known. Thus the quest remains open for an
integrated approach to reconstruction from $\sd$-quantized compressed
sensing measurements that combines numerical tractability and
guaranteed recovery.

\section{Scalar Quantization and Compressive Sensing}
\label{sec:scalarCS}

The interplay of scalar quantization and compressed sensing has been
widely explored in the literature. In addition to the lower bounds
discussed in~\ref{sec:perf-bounds-sparse}, there is significant
interest in providing practical quantization schemes and
reconstruction algorithms with strong performance guarantees.

This part explores these results. Our development considers
the following quantized compressed sensing (QCS) model:
\begin{equation}
  \label{eq:qcs-pcm}
  q = \qQop(y) = \qQop(A x),  
\end{equation}
where $x \in \bb R^n$ and $A \in \bb R^{m \times n}$. The sensing
matrix can be, for instance, a random Gaussian sensing matrix $A$ such
that $a_{ij} \sim_{\rm iid} \cl N(0,1)$. Note that the scaling of the
  entries of the sensing matrix should be independent of $m$. This
  allows us to fix the design of the scalar quantizer $\qQop$ since
  the dynamic range of the components of $Ax$ is then independent of
  the number of measurements. This has no consequence on some of the
  common requirements the sensing matrix must satisfy, such as the
  Restricted Isometry Property (see Chap. 1), as soon as an
  appropriate rescaling of $A$ is applied. For instance, if $A$ has
  RIP of order $2k$ and if $A \to \lambda A$ for some $\lambda>0$,
  then $A/\lambda$ has RIP of the same order and the error bound
  (1.20) in the stability Theorem 1.6 remains
  unchanged~\cite{Jacques2010}.

The first two parts, Sec.~\ref{sec:unif-scal-quant} and
Sec.~\ref{sec:nonunif-scal-quant}, focus on the \emph{high resolution
  assumption} (HRA) that simplifies the QCS model. Under HRA, the
quantization bin widths---$\qD$ or the distance between two
consecutive thresholds---are small with respect to the dynamic range
of the unquantized input. This allows us to model the quantization
distortion $\qQop(A x) - Ax$ as uniform white noise
\cite{gray1998quantization}. Determining bounds on its power and
moments can better constrain signal reconstruction methods, such as the
basis pursuit denoise (BPDN) program
\cite{Chen98atomic,candes2008rip}, which is commonly used for
reconstructing signals whose CS measurements are corrupted by
Gaussian noise. However, the price to pay is an oversampling in CS
measurements.

Sec.~\ref{sec:unif-scal-quant-with-saturation} considers scalar
quantizers with \emph{saturation}. Saturation induces information loss
in the measurements exceeding the saturation level. However,
democracy---a key property of compressive sensing measurements that
makes every measurement equally informative---provides robustness
against such corruption.

In Sec.~\ref{sec:1bitCS}, very low-resolution quantization is studied
through 1-bit compressed sensing.  In this case, the HRA cannot be
assumed anymore---the quantization bins are the two semi-infinite
halves of the real line---and the analysis of the QCS model relies on
high dimensional geometric arguments.

Finally, Sec.~\ref{sec:extensions} studies how noise, either on the
signal or on the measurements, can impact the QCS model
\eqref{eq:qcs-pcm}, the reconstruction error and the quantizer
trade-offs. In particular, at constant bit budget $\qR = m\qB$, the
total noise power determines the optimal trade-off between quantizer
precision and number of measurements.

\subsection{Uniform Scalar Quantization}  
\label{sec:unif-scal-quant}

First we consider the QCS model \eqref{eq:qcs-pcm} using a uniform
quantizer with resolution $\qD$ and a set of levels \qQal,
$$
q = \qQop(y) = \qQop(A x) \in \qQal^m,
$$ 
measuring a signal $x\in \bb R^n$ using a sensing matrix $A\in\bb
R^{m\times n}$. For simplicity, we assume henceforth that $x$ is
sparse in the canonical basis, \ie $\Psi = \Id$.

We consider a quantizer $\qQop$ that has uniform quantization regions, \ie
$t_{j+1}-t_j = \qD$ for all $j$, and, setting $t_j = j \qD$,
quantization levels
$\qlev_j = \frac{t_j+t_{j+1}}{2} = (j +\tinv{2})
\qD$ in $\qQal$. 

By definition, the signal $x$ satisfies the following \emph{quantization consistency}
constraint (QC$_u$)
\begin{equation}
  \label{eq:qc-unfi}
  \|\qq - Ax\|_{\infty} \leq \qD/2. \tag{QC$_{\rm u}$}
\end{equation}
From this fact, we can also deduce that 
$$
\|A x - q\|_2 \leq \sqrt{m}\,\|A x - q\|_\infty
\leq \sqrt m \qD/2.
$$ 
This shows that the QCS model can be assimilated to a noisy CS model
\begin{equation}
  \label{eq:noisy-unfi-QCS}
  q = \qQop(A x) = Ax + \noise,
\end{equation}
with a ``noise'' $\noise = \qQop(A x) - Ax$ of bounded $\ell_2$-norm, \ie
$\|\noise\|_2 \leq \sqrt m \qD/2$. 

The quantization noise power can be further reduced using the high
resolution assumption. Under this assumption, the coefficients of $y$
may lie anywhere in the quantization region determined by the
coefficients of $q$ and it is natural to model the quantization
distortion $\noise$ as a uniform white noise, \ie
$$
\noise_i \sim_{\rm iid}\ \cl
U([-\qD/2,\qD/2]).
$$ 
Under this model, a simple use of the
Chernoff-Hoeffding bound \cite{hoeffding1963pis} provides, with high probability
$$
\|\noise\|^2_2 \leq \epsilon^2_2 := \tfrac{\qD^2}{12} m + \zeta
\tfrac{\qD^2}{6\sqrt 5} m^{1/2},
$$
for a small constant $\zeta>0$.

The first approach in modeling and understanding QCS exploited this
bound and the development of noise-robust CS approaches to impose a
\emph{distortion consistency constraint} (DC$_u$)~\cite{candes2006ssr}
\begin{equation}
  \label{eq:dc-unfi}
  \|\qq - Ax'\|_2 \leq \epsilon_2, \tag{DC$_{\rm u}$}
\end{equation}
on any candidate signal $x'$ estimating $x$. This was indeed a natural
constraint to consider since most noise-robust compressed sensing
reconstruction methods can incorporate a bounded $\ell_2$-norm
distortion on the measurements. For instance, the BPDN program can
find a solution $\hat x$ of
\begin{equation}
  \label{eq:BPDN}
\hat x\ =\ \arg\,\min_{z} \|z\|_1\ \st\ \|q - Az\|_2 \leq \epsilon_2. \tag{BPDN}  
\end{equation}
Then, if the sensing matrix $A' = A/\sqrt{m}$ satisfies the RIP with constant
$\delta \leq 1/\sqrt{2}$ on $2k$ sparse signals, it is known \cite{Cai14} that
$$
\|x - \hat x\|_2 \lesssim \tinv{\sqrt m}\,\epsilon_2 + \tinv{\sqrt k}\,\sigma_k(x)_1 \asymp \qD + \tinv{\sqrt k}\,\sigma_k(x)_1,
$$
where $\sigma_k(x)_1$ is the best $k$-term approximation
  defined in (1.2).
  
This approach has two drawbacks. First, there is no guarantee that the
solution $\hat x$ satisfies the QC$_{\rm u}$ constraint above, \ie
$\|\qq - A\hat x\|_{\infty} \nleq \qD/2$. This shows that some sensing
information has been lost in the reconstruction. Moreover, as
described in Sec.~\ref{sec:scal-quant-overs}, the consistency of the
solution helps in reaching the lower bound
\cite{goyal_1998_lowerbound_qc,powell_consistent,jacques2014error}
\[ (\bb E \|x - \hat x\|^2)^{1/2} \gtrsim \tfrac{k}{m}\,\qD
\] 
in the oversampled setting.  Second, from a maximum \emph{a posteriori}
standpoint, since every constrained optimisation corresponds to an
unconstrained Lagrangian formulation, imposing a small $\ell_2$-norm
on the residual $q - A \hat x$ can be viewed as enforcing a Gaussian
distribution on $\noise$, which is not the uniform one expected from
the HRA.

To circumvent these two limitations, \cite{Jacques2010} studied the
Basis Pursuit DeQuantizer (BPDQ$_p$) program
\begin{equation}
  \label{eq:BPDQ}
\hat x_p = \arg\,\min_{z}\ \|z\|_1\ \st\ \|q - Az\|_p \leq \epsilon_p, \tag{BPDQ$_p$}
\end{equation}
where $\epsilon_p$ must be carefully selected in order for $x$ to be a
feasible point of this new $\ell_p$-constraint. If $\epsilon_p\to \qD$
as $p\to\infty$, the BPDQ$_p$ solution $\hat x_p$ tends to be
consistent with the quantized measurements. But what is the price to
pay, \eg in terms of number of measurements, for being allowed to
increase $p$ beyond $2$?

To answer this, we need a variant of the restricted isometry property.
\begin{definition}
  Given two normed spaces $\mathcal{X}=(\bb R^m,\|\cdot\|_{\cl X})$
  and $\mathcal{Y}=(\bb R^n,\|\cdot\|_{\cl Y})$ (with $m< n$), a matrix
  $A\in\bb R^{m\times n}$ has the Restricted Isometry
  Property from $\mathcal{X}$ to $\mathcal{Y}$ at order
  $k\in\bb N$, radius $0\leq\delta<1$ and for a normalization
  $\mu>0$, if for all $x\in\Sigma_k := \{ u \in \bb R^N: \|u\|_0 \leq k\}$,
\begin{equation}
  \label{eq:rip-p}
  (1-\delta)^{1/\kappa}\,\|x\|_{\mathcal{Y}} \leq \tinv{\mu} \|A x\|_{\mathcal{X}} \leq
  (1+\delta)^{1/\kappa}\,\|x\|_{\mathcal{Y}},
\end{equation}
the exponent $\kappa$ depending on the spaces $\cl X$ and $\cl Y$. To lighten notation, we write that
$A$ is RIP$_{\mathcal{X},\mathcal{Y}}(k,\delta,\mu)$.
\end{definition}

In this general definition, the common RIP is equivalent to
RIP$_{\ell_2^m,\ell_2^n}(k,\delta,1)$ with $\kappa=2$ (see Chap. 1, Eq.
(1.10)). Moreover, the RIP$_{p,k,\delta'}$ defined in \cite{Berinde2008}
is equivalent to the RIP$_{\ell_p^m,\ell_p^n}(k,\delta,\mu)$ with
$\kappa=1$, $\delta'=2\delta/(1-\delta)$ and $\mu=1/(1-\delta)$. Finally, the
Restricted $p$-Isometry Property proposed in
\cite{chartrand2008restricted} is also equivalent to the
RIP$_{\ell_p^m,\ell_2^n}(k,\delta,1)$ with $\kappa=p$.

To characterize the stability of BPDQ we consider the space $\cl X =
\ell_p^m := (\bb R^m, \|\cdot\|_p)$ and $\cl Y = \ell_2^n := (\bb R^m,
\|\cdot\|_2)$ with $\kappa=1$, and we write RIP$_{p}$ as a shorthand
for RIP$_{\ell_p^m,\ell_2^n}$. At first sight, it could seem
  unnatural to define an embedding of $\cl X=\ell^n_p$ in $\cl Y =\ell^m_2$ for $p \neq 2$,
those spaces being not isometrically isomorphic to each other for $m = n$. However,
the RIP$_p$ rather sustains the possibility of an isometry between
$\cl X
\cap A\Sigma_k$ and $\cl Y\cap\Sigma_k$. We will see in
Prop.~\ref{prop:grip-gauss} that the existence of such a relation comes with
an exponential growth of $m$ as $p$ increases, a phenomenon that can
be related to Dvoretsky's theorem when specialized to those Banach
spaces \cite{ledoux}.

From this new characterization,
one can prove the following result.
\begin{theorem}[\cite{Jacques2010,Jacques2013}]
  \label{prop:l2-l1-instance-optimality-BPDQ}
  Let $k\geq 0$, $2\leq p <\infty$ and $A\in \bb R^{m\times n}$ be a RIP$_{p}(s,\delta_s,\mu_p)$ matrix for
  $s\in\{k,2k,3k\}$ and some normalization constant $\mu_p>0$. If
  \begin{equation}
    \label{eq:cond-on-delta-p}
    \delta_{2k}+\sqrt{(1+\delta_k)(\delta_{2k}+\delta_{3k})(p - 1)} < 1/3,   
  \end{equation}
then, for any signal $x\in\bb R^n$ observed according to the noisy
sensing model $y=A x + n$ with $\|n\|_{p}\leq
  \epsilon_p$, the unique solution $\hat x_p$  obeys
\begin{equation}
\label{eq:BPDQ-l2-l1-inst_opt}
\|x^* - x\| \ \leq\ 4\,\tinv{\sqrt k}\,\sigma_k(x)_1\ +\ 8\,\epsilon_p/\mu_p,
\end{equation}
where, again, $\sigma_k(x)_1$ denotes the best $k$-term approximation.
\end{theorem}
This theorem follows by generalizing the fundamental result of
Cand\`es in \cite{candes2008rip} to the particular geometry of Banach
spaces $\ell_p^m$. It shows that, if $A$ is RIP$_p$ with particular
requirement on the RIP$_p$ constant, the BPDQ$_p$ program is stable
under both measurement noise corruption and departure from the strict
sparsity model, as measured by $e_0$. In particular, under the same
conditions, given a measurement noise $\noise$ and some upper bounds
$\epsilon_p$ on its $\ell_p$-norm, \eqref{eq:BPDQ-l2-l1-inst_opt}
provides the freedom to find the value of $p$ that minimizes
$\epsilon_p/\mu_p$.

This is exactly how QCS signal recovery works. Following Theorem
\ref{prop:l2-l1-instance-optimality-BPDQ} and its stability result
\eqref{eq:BPDQ-l2-l1-inst_opt}, we jointly determine a RIP$_p$ sensing
matrix with known value $\mu_p$ and a tight error bound $\epsilon_p$
on the $\ell_p$ norm of the residual $q-Ax$ under HRA. The existence
of a RIP$_p$ matrix is guaranteed by the following
result~\cite{Jacques2010,Jacques2013}.
\begin{proposition}[\bf RIP$_{p}$ Matrix Existence]
\label{prop:grip-gauss}
Let a random Gaussian sensing matrix $A\in \bb R^{m\times n}$ be such
that $a_{ij} \sim_{\rm iid} \cl N(0,1)$, $p\geq 1$ and $0\leq \eta <
1$. Then, $A$ is RIP$_{p}(k,\delta_k,\mu_p)$ with probability higher than
$1-\eta$ when we have jointly $m\geq 2^{p+1}$ and 
\begin{align}
\label{eq:SGR-RIP-measur-bound}
m \geq
m_0^{\max(p/2,1)}\quad \text{with}\ m_0 = O(\delta_k^{-2}\,
\big(k\log(\tfrac{n}{k}) + k \log(\delta_k^{-1}) + \log\tfrac{2}{\eta}\big)).
\end{align}
Moreover, $\mu_p = \Theta(m^{1/p} \sqrt{p+1})$.
\end{proposition}
There is thus an exponential price to pay
for a matrix $A$ to be RIP$_p$ as $p$ increases: roughly speaking, for
$p\geq 2$, we
need $m \geq m^{p/2}_0 = O( k^{p/2} \log^{p/2} (n/k))$ measurements for
satisfying this property with non-zero probability. 

To estimate a tight value of $\epsilon_p$ in the case of quantization
noise---since, under HRA $\noise_j \sim_{\rm iid} \cl
U([-\qD/2,\qD/2])$---we can show that
\begin{equation}
\label{eq:epsilon-p-def}
\|\noise\|_p \leq \epsilon_p\ :=\  
\tfrac{\qD}{2\,(p+1)^{1/p}}\,
\big(\,m + \zeta\,(p+1)\,\sqrt{m}\,\big)^{\inv{p}},
\end{equation}
with probability higher than $1-e^{-2\zeta^2}$. Actually, for $\zeta =
2$, $x$ is a feasible solution of the
BPDQ$_p$ fidelity constraint with a probability exceeding $1-e^{-8} > 1-
3.4 \times 10^{-4}$.  

Finally, combining the estimation $\epsilon_p$ with the bound on
$\mu_p$, we find, under the conditions of Prop.~\ref{prop:grip-gauss}, 
\begin{equation}
  \label{eq:unif-error-decay}
  \tfrac{\epsilon_p}{\mu_p} \lesssim \tfrac{\qD}{\sqrt{p+1}}. 
\end{equation}
This shows that, in the high \emph{oversampled sensing scenario}
driven by \eqref{eq:SGR-RIP-measur-bound}, and provided the RIP$_p$ constants $\{\delta_k,\delta_{2k},\delta_{3k}\}$ satisfy
\eqref{eq:cond-on-delta-p}, the part of the reconstruction error due to
quantization noise behaves as $O(\qD/\sqrt{p+1})$.
This is also the
error we get if $x$ is exactly $k$-sparse since then $e_0$ vanishes in \eqref{eq:BPDQ-l2-l1-inst_opt}.

If we solve for $p$, we can see that the error decays as $O(\qD/\sqrt{\log m})$
as $m$ increases. There is possibly some room
for improvements since, as explained in
Sec.~\ref{sec:scal-quant-overs}, the lower bound on reconstruction of
sparse signal is $\Omega(\qD/m)$.
Beyond scalar quantization schemes,
Sec.~\ref{sec:SDCS} will also show that much better theoretical error
reduction can be expected using \sd quantization.

\begin{figure*}[t]
  \centering
  \null\hfill
  \subfigure[\label{fig:first-exper-qual}]
  {\raisebox{-3mm}{\includegraphics[width=.35\textwidth]{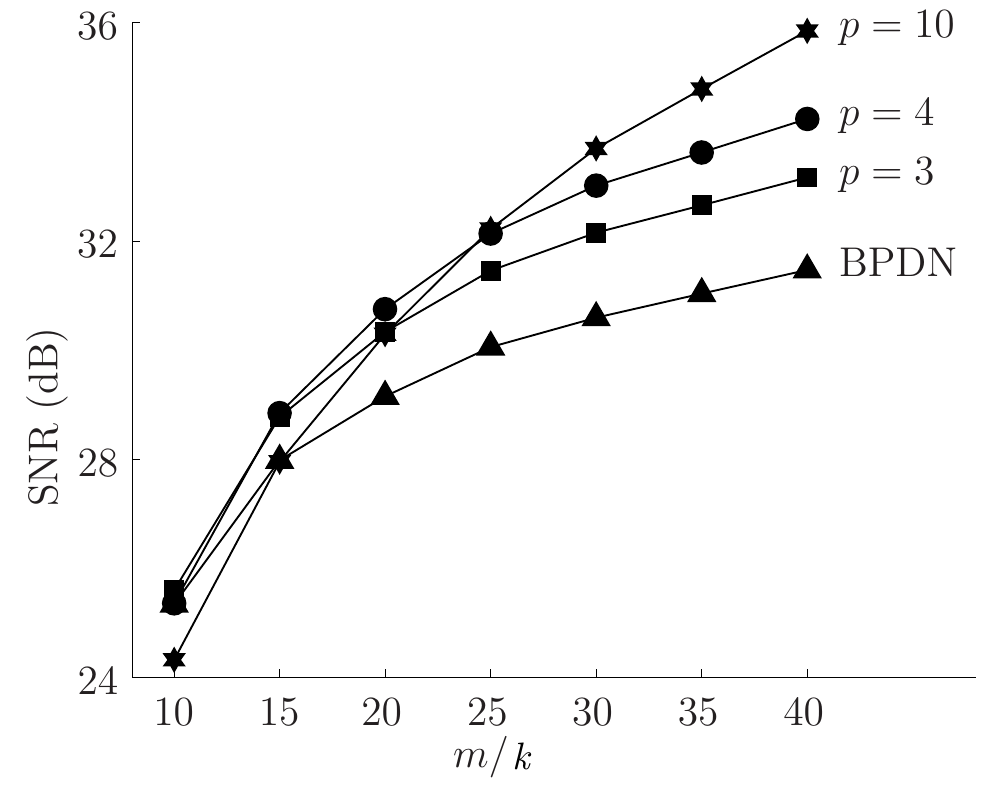}}}
  \hfill
  \subfigure[\label{fig:nrh-1}]
  {\includegraphics[width=.25\textwidth]{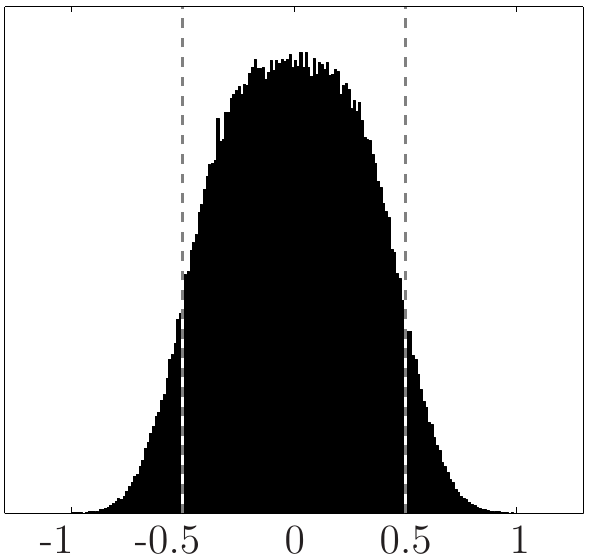}}
  \hfill
  \subfigure[\label{fig:nrh-2}]
  {\includegraphics[width=.25\textwidth]{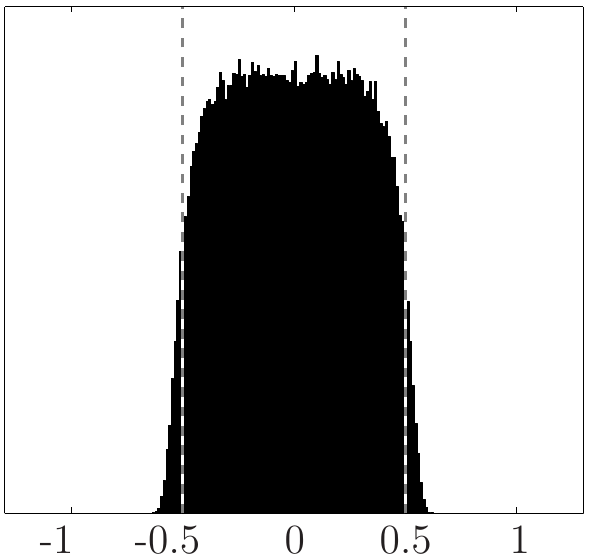}}  
  \hfill\null
  \caption{(a) Quality of BPDQ$_p$ for different $m/k$ and
    $p$. (b) and (c): Histograms of $\qD^{-1}(A \hat{x}-q)_i$ for $p=2$ and
  for $p=10$, respectively. \label{fig:first-exper}}
\end{figure*}

Interestingly, we can, however, observe a numerical gain
in using BPDQ$_p$ for increasing values of $p$ when the signal $x$ is observed by the model
\eqref{eq:noisy-unfi-QCS} and when $m$ increases beyond the minimal
value $m_0$ needed for stabilizing BPDN (\ie BPDQ$_2$).  

This gain is depicted in Fig.~\ref{fig:first-exper}. The plots on the
left correspond to the reconstruction quality, \ie the value ${\rm
  SNR} = 20 \log(\|x\|/\|x - \hat x_p\|)$ expressed in dB, reached by
BPDQ$_p$ for different values of $p$ and $m/k$. The original signal
$x$ has dimension $n=1024$ and is $k$-sparse in the canonical basis,
with support of size $k=16$ uniformly random and normally distributed
non-zero coefficients. Each point of each curve represents average
quality over 500 trials.  For each sparse signal $x$, $m$ quantized
measurements were recorded using \eqref{eq:noisy-unfi-QCS} with a
random Gaussian sensing matrix $A$ and $\qD = \|A x\|_\infty/40$. The
reconstruction was done by solving BPDQ$_p$ with the Douglas-Rachford
algorithms \cite{Jacques2010}, an efficient convex optimization method
solving constrained programs, such as BPDQ\footnote{The code of BPDQ
  is freely available at \url{http://wiki.epfl.ch/bpdq}.}, using
simpler \emph{proximal operators} \cite{combettes2011proximal}.
Fig.~\ref{fig:first-exper-qual} shows that higher oversampling ratio
$m/k$ allows the use of higher $p$ with significant gain in the
reconstruction quality. However, if $m/k$ is low, \ie close to
$m/k=10$, the best quality is still reached by BPDN. The quantization
consistency of the reconstruction, \ie the original motivation for
introducing the BPDQ$_p$ program, can also be tested. This is shown on
Fig.~\ref{fig:nrh-1} and Fig.~\ref{fig:nrh-2} where the histograms of
the components of $\qD^{-1}(A \hat x_p - q)$ are represented for $p=2$
and $p=10$ at $m/k=40$. This histogram for $p=10$ is indeed closer to
a uniform distribution over $[-1/2,1/2]$, while the one at $p=2$ is
mainly Gaussian.
\subsection{Non-Uniform Scalar Quantization}  
\label{sec:nonunif-scal-quant}

If the distribution of the measurements is known,
  quantization distortion can be decreased by adopting a non-uniform
  scalar quantizer.
For instance, when $A$ is a random Gaussian matrix
viewing the signal as fixed and the matrix as randomly drawn, the
distribution of the components of $y=A x$ is also Gaussian with a
variance proportional to the signal energy
$\|x\|_2^2$ (and similarly, for other matrix constructions, such as
  ones drawn with random sub-Gaussian entries). Assuming the acquired
  signal energy can be fixed, \eg using some automatic gain control, the
  known distribution of the measurements can be exploited in the
  design of the quantizer, thanks for example to the Lloyd-Max algorithm
  mentioned in Sec.~\ref{sec:intro_scalar}~\cite{Lloyd1982}. In
  particular, the
  quantization thresholds and levels are then optimally adjusted to this
  distribution.

This section shows that the formalism developed in Sec.~\ref{sec:unif-scal-quant} can indeed
be adapted to non-uniform scalar quantizer. To understand this adaptation, we exploit a common tool in
quantization theory~\cite{gray1998quantization}: any non-uniform
quantizer can be factored as the composition of a ``compression'' of
the real line over $[0, 1]$ followed by a uniform quantization of the
result that is finally re-expanded on $\bb R$. Mathematically,
\begin{equation}
  \label{eq:compander}
  \qQop\ =\ \cl G^{-1} \circ \qQop_{\qD} \circ \cl G, 
\end{equation}
where $\cl G:\bb R \to [0,1]$ is the \emph{comp}ressor and $\cl
G^{-1}:[0,1]\to \bb R$ is the exp\emph{ander}, giving the name
\emph{compander} as a portemanteau.

In particular, under HRA, the compressor $\cl G$ of a distortion
optimal quantizer, \ie one that minimizes $\bb E|X-\qQop(X)|^2$ for a
source modeled as a random variable $X$ with pdf $\pdf$, must satisfy  
$$
\tfrac{\ud}{\ud \lambda} \cl G(\lambda) = \left( \int \pdf^{1/3}(t) \ud t\right)^{-1} \pdf^{1/3}(\lambda),
$$ 
and if $\qQop$ is an
optimal $B$-bit quantizer (\eg obtained by Lloyd-Max method) then $\qD
= 2^{-B}$ in \eqref{eq:compander}. In this case, the Panter and Dite formula estimates the
quantizer distortion as \cite{panter1951quantization}
$$
\bb E|X-\qQop(X)|^2 \simeq_B \tfrac{2^{-2 B}}{12} \fnorm{\pdf}_{1/3}
=: \sigma^2_{\rm PD}, 
$$
with $L_s$-norm $\fnorm{\pdf}_{s} = (\int |\pdf^{s}(t)| \ud t)^{1/s}$ and 
where ``$\simeq_B$'' means that the relation tends to an equality when
$B$ is large. The rest of this section assumes that the expected distribution is Gaussian, \ie if $\pdf
\sim \cl
N(0,\sigma_0^2)$ and $\fnorm{\pdf}_{1/3} = \tinv{2}\sqrt{3}\pi
\sigma^2_0$, as it comes by seeing the signal fixed (with known
  energy) and the Gaussian
matrix random in~CS. 

Compander theory generalizes quantization
consistency in the ``compressed'' domain, \ie
$$
|\cl G(\lambda) - \cl G(\qQop(\lambda))| \leq \qD/2 = 2^{-B-1}.
$$ Therefore, for the right compressor $\cl G$, in the noiseless QCS
model \eqref{eq:qcs-pcm}, the signal $x$ provides consistency
constraints to be imposed on any reconstruction candidate $x'$:
\begin{equation}
  \label{eq:quant-const-nonunif}
  \|\cl G(A x') - \cl G(q)\|_\infty \leq \qD/2 = 2^{-B-1}. \tag{QC}  
\end{equation}
This generalizes the uniform quantization consistency \eqref{eq:qc-unfi} introduced
in Sec.~\ref{sec:unif-scal-quant}.

The compander formalism is leveraged in \cite{Jacques2013}, to
generalize the approach described in Sec. \ref{sec:unif-scal-quant} to
non-uniform quantization. In particular, a new set of parametric
constraints are introduced, the $p$-Distortion Consistency (or D$_p$C)
for $p\geq 2$. These have for limit cases the QC above and the
\emph{distortion consistency} constraint (DC) arising from Panter and
Dite formula, namely, the constraint imposing any reconstruction
candidate $x'$ to satisfy \cite{Dai2009}
\begin{equation}
  \|A x' - q\|^2_2 \leq \epsilon^2_{\rm PD} := m \sigma^2_{\rm PD}, \tag{DC}
\end{equation}
with DC asymptotically satisfied by~$x$ when both $B$ and
$m$ are large.

The D$_p$C constraint corresponds to imposing that a candidate
signal $x'$ satisfies 
\begin{equation}
\label{eq:dpc}
\|A x' - \qQop_p[q]\|_{p,w}\ =\ \|A x' - \qQop_p[Ax]\|_{p,w}\ \leq\ \epsilon_{p,w}, \tag{D$_p$C}   
\end{equation}
where $\|v\|_{p,w} = \|\diag(w) v\|_p$ is the weighted $\ell_p$-norm 
of $v \in \bb R^m$ with weights $w \in
\bb R^m_+$, denoting by $\diag(w)$ the diagonal matrix having $w$ on its diagonal. The mapping $\qQop_p:\bb R^m \to \bb R^m$ is a
  post-quantization modification of $q$ characterized
componentwise hereafter and such that $\qQop_p[q]=\qQop_p[Ax]$. 

Under HRA, a careful design of $\qQop_p$, $w$ and the bounds
$\epsilon_{p,w}$ ensures that D$_2$C amounts to imposing DC on $x'$
and, that as $p\to +\infty$, D$_p$C tends to QC \cite{Jacques2013}.
Briefly, if $q_i$ falls in the quantization bin $\qcell_j$,
$\qQop_p(q_i)$ is defined as the minimizer of
$$
\min_{\lambda\in\qcell_j}\ \int_{\qcell_j} |t - \lambda|^p\ \pdf(t)\,\ud t.
$$
Actually, $\qQop_2(q_i) = q_i$ by equivalence with \eqref{eq:LM-condition}, and $\lim_{p\to \infty}\qQop_p(q_i) =
\tinv{2}(t_j+t_{j+1})$. The weights are defined by the
quantizer compressor $\cl G$ with $w_i(p)= \tfrac{\ud}{\ud \lambda} \cl G\big(\qQop_p[q_i]\big)^{\frac{p-2}{p}}$. 
Moreover, under HRA and asymptotically in $m$, an optimal bound $\epsilon_p$ reads
$\epsilon_{p,w}^p = m\,
\tfrac{2^{-Bp}}{(p+1)\,2^{p}}\,\fnorm{\pdf}_{1/3}$. For $p=2$, $\epsilon_{2,w}=\epsilon_{\rm PD}$ matches the distortion power estimated by
the Panter and Dite formula, while for $p\to
+\infty$, $\epsilon_{p,w} \to \tinv{2} 2^{-B}$, \ie half the size of
the uniform quantization bins in the domain compressed by $\cl G$.   

Similarly to Sec.~\ref{sec:unif-scal-quant}, using
\eqref{eq:dpc} as a fidelity constraint in the signal reconstruction leads to the definition of a Generalized
Basis Pursuit DeNoise program:
\begin{equation}
  \hat{x}_{p, w}\ =\ \arg\,\min_{z\,\in\,\bb R^n} \|z\|_1
  \ {\rm
  s.t.}\ \|\qQop_p(q) - A z\|_{p,w} \leq \epsilon_{p,w}. \tag{GBPDN$(\ell_{p,w})$}
\end{equation}

Ideally, we would like to directly set $p=\infty$ in order to enforce
consistency of $\hat{x}_{p, w}$ with $q$. However, as studied in
\cite{Jacques2013}, it is not certain that this limit case
  minimizes the reconstruction error $\|x - \hat{x}_{p, w}\|$ as a
  function of $p$, given a certain number of measurements $m$.

Actually, the stability of GBPDN can be established from the one
of BPDQ (Sec.~\ref{sec:unif-scal-quant}) if we impose $A$ to satisfy the more general RIP$_{\ell^m_{p,w},
  \ell^n_2}$, as formally defined in \eqref{eq:rip-p}. Indeed, for any
weighting vector $w$, we have always $\|\qQop_p(q) - A z\|_{p,w} =
\|q' - A' z\|_{p}$ with $q' = \diag(w) \qQop_p(q)$ and $A' = \diag(w)
A$. Therefore, we know from
Theorem~\ref{prop:l2-l1-instance-optimality-BPDQ} that if $A'$ is
RIP$_p$, or equivalently if $A$ is RIP$_{\ell^m_{p,w},
  \ell^n_2}$, with the additional condition \eqref{eq:cond-on-delta-p} on its RIP constants at different
sparsity levels, then the solution of GBPDN$(\ell_{p,w})$ will be stable in the
sense of \eqref{eq:BPDQ-l2-l1-inst_opt}, \ie
$$
\|\hat x_{p,w} - x\|\ \lesssim\ \tfrac{\epsilon_{p,w}}{\mu_{p,w}} +
\tfrac{\sigma_k(x)_1}{\sqrt k}.
$$

Compared to the unit weights case (as involved by
the RIP$_p$), a random Gaussian matrix $A$ with $a_{ij} \sim_{\rm iid}
\cl N(0,1)$ satisfies the RIP$_{\ell^m_{p,w},
  \ell^n_2}(k,\delta_k,\mu_{p,w})$ with high probability provided that $m$ grows like
$O\big( (\theta_p \delta_k^{-2} (k
\log(n/k))^{p/2} \big)$.
The ratio $\theta_p := \|w\|_\infty/(m^{-1/p}\|w\|_p)$ depends on the \emph{conditioning} of $w$. It is equal to 1 for constant
weights (recovering \eqref{eq:SGR-RIP-measur-bound}), while it
increases with the dynamic range of $w$. For the weight $w(p)$ defined previously and with a Gaussian
optimal quantizer, $\theta^{p/2}_p
\simeq_{m,B} \sqrt{p+1}$ asymptotically in $m$ and $B$.

As for the uniform case, a strong (polynomial) oversampling in
$m$ is thus required for satisfying the RIP$_{\ell^m_{p,w},
  \ell^n_2}$ at $p>2$ compared to the minimal number of measurements needed
at $p=2$. However, an asymptotic analysis of
$\epsilon_{p,w}/\mu_{p,w}$ shows that the GBPDN reconstruction error due to quantization for a Gaussian
sensing matrix behaves like \cite{Jacques2013}
$$
\|\hat x_{p,w} - x\| \lesssim \tfrac{2^{-B}}{\sqrt{p+1}} + \tfrac{\sigma_k(x)_1}{\sqrt k},
$$ 
This error decay is thus similar to the one found in
\eqref{eq:unif-error-decay} for uniform QCS with
now a direct interpretation in terms of the quantizer bit-depth $B$.

Efficient convex optimization methods, like those relying on proximal
algorithms \cite{combettes2011proximal}, can also be used to
numerically solve GBPDN.  In \cite{Jacques2013}, numerical simulations
show that the reconstruction qualities reached in the reconstruction
of sparse signals from their non-uniformly quantized measurements
behave similarly, with respect to $p$ and $m$, to those observed in
Sec.~\ref{sec:unif-scal-quant} for the uniformly quantized CS setting.

We should also remark that beyond QCS, the stability of GBPDN (when
$A$ is RIP$_{p,w}$) can also be used for reconstructing signals
acquired under a (heteroscedastic) noisy sensing model $y = Ax +
\noise$ where $\noise\in \bb R^m$ is an additive generalized Gaussian
noise with bounded $\ell_{p,w}$-norm for some specific weight $w\in
\bb R^m_+$ \cite{varanasi1989pgg,Jacques2013}.

\subsection{Finite-Range Scalar Quantizer Design}  
\label{sec:unif-scal-quant-with-saturation}

So far we have only considered a scalar quantizer model without
saturation. Practical scalar quantizers have a finite range, which
implies a saturation level $\pm\qS$ and, using \qB\ bits per
coefficient, a quantization interval equal to
\begin{align}
  \qD=\qS 2^{-B+1}.
  \label{eq:quant_interval}
\end{align}
In order to determine the optimal saturation rate, the system designed
needs to balance the loss of information due to saturation, as
\qS\ decreases, with the increased quantization error due to an
increasing quantization interval in~\eqref{eq:quant_interval}, as
\qS\ increases. In classical systems, this balance requires setting
the quantization level relatively close to the signal amplitude to
avoid saturation. On the other hand, in compressive sensing systems,
the incoherence of the measurements with the sparsity basis of the
signal makes them more robust to loss of information and enables
higher saturation levels with smaller quantization
intervals.

A key property of compressive measurements, which provides the
robustness to loss of information, is {\em democracy}. Intuitively,
each measurement contributes an equal amount of information to the
reconstruction. If the signal is slightly oversampled, relative to the
rate required for CS reconstruction, then any subset with enough measurements should be sufficient to recover the
signal. The notion of democracy was first introduced
in~\cite{CalDau::2002::The-pros-and-cons,Gun::2000::Harmonic-analysis}
in the context of information carried in each bit of the
representation; the definition below strengthens the concept and
formulates it in the context of compressive
sensing~\cite{DLBB_ECETR09,LasBouDav::2009::Demcracy-in-action}.

\begin{definition}
Let $A\in{\mathbb R}^{m \times n}$, and let $\widetilde{m} \le m$ be
given.  We say that $A$ is $(\widetilde{m},k,\delta_k)$-democratic if,
for all row index sets $\Gamma$ such that $|\Gamma|\ge \widetilde{m}$,
any matrix $\widetilde{A}=((A^T)_\Gamma)^T$, \ie comprised of a
$\Gamma$-subset of the rows of $A$, satisfies the RIP of order $k$
with constant $\delta_k$.
\end{definition}
This definition takes an adversarial view of democracy: a matrix $A$
is democratic if an adversary can pick any $d=m-\widetilde{m}$ rows to
remove from $A$, and the remaining matrix still satisfies the
RIP. This is a much stronger guarantee than just randomly selecting a
subset of the rows to be removed. Such a guarantee is important in the
case of saturation robustness because the saturated measurements are
the largest ones in magnitude, \ie potentially the ones most aligned with the
measured signal and, presumably, the ones that capture a significant
amount of information. Still, despite this strict requirement,
randomly generated matrices can be democratic if they have a
sufficient number of rows.

\begin{theorem}[\cite{DLBB_ECETR09}] \label{thm:democracy}
Let $A\in{\mathbb R}^{m \times n}$ with elements $a_{ij}$ drawn
according to $\mathcal{N}(0, \frac{1}{m})$ and let $\widetilde{m} \le m$, $k <
\widetilde{m}$, and $\delta \in (0,1)$ be given.  Define $d = m -
\widetilde{m}$. If
\begin{align}
\label{eq:Mdem}
m = C_{1} (k + d)\log\left(\frac{n+m}{k + d}\right),
\end{align}
then with probability exceeding $1 - 3 e^{-C_2 m}$ we have that $A$ is
$(\widetilde{m},k,\delta/(1-\delta))$-democratic, where $C_1$ is
arbitrary and $C_2 = (\delta/8)^2 - \log(42 e/ \delta)/C_1.$
\end{theorem}
The practical implication of democratic measurements is that
information loss due to saturated measurements can be tolerated. 

Saturated measurements are straightforward to detect, since they
quantize to the highest or the lowest level of the quantizer. The
simplest approach is to treat saturated measurements as corrupted, and
reject them from the reconstruction, together with the corresponding
rows of $A$. As long as the number of saturated measurements is not
that large, the RIP still holds and reconstruction is possible using
any sparse reconstruction algorithm.

However, saturated measurements do contain the information that the
measurement is large. In the context of consistent reconstruction,
they can be used as constraints in the reconstruction process. If a
measurement $i$ is positively saturated, then we know that $(Ax)_i \ge \qS-\qD$. Similarly, if it is negatively saturated,
$(Ax)_i \le -\qS+\qD$. These constraints can be imposed on
any reconstruction algorithm to improve performance~\cite{LasBouDav::2009::Demcracy-in-action}.

\begin{figure}[!t] 
   \centering
   \begin{tabular}{ccc}
  \includegraphics[width=.32\linewidth]{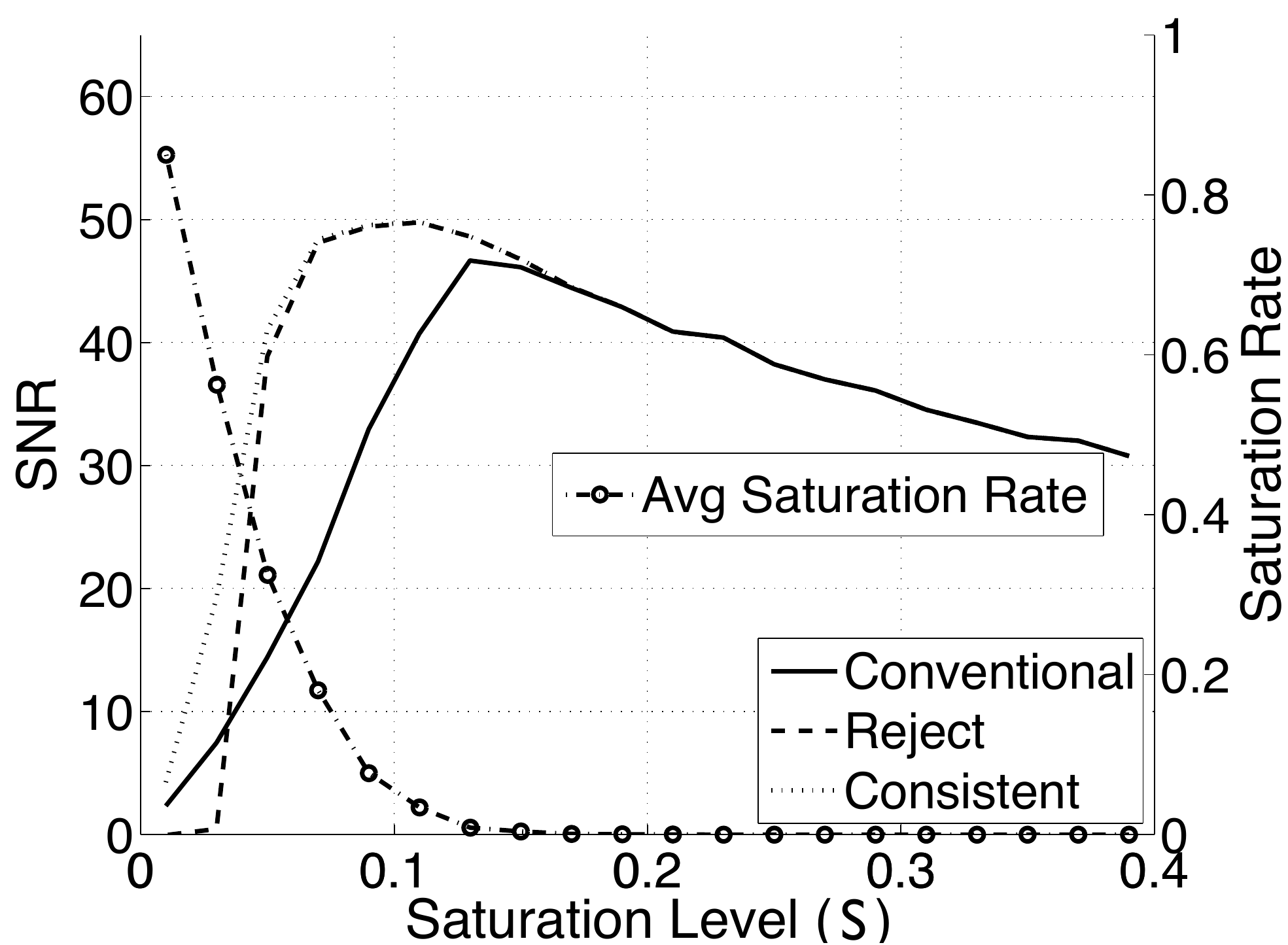}&
  \includegraphics[width=.32\linewidth]{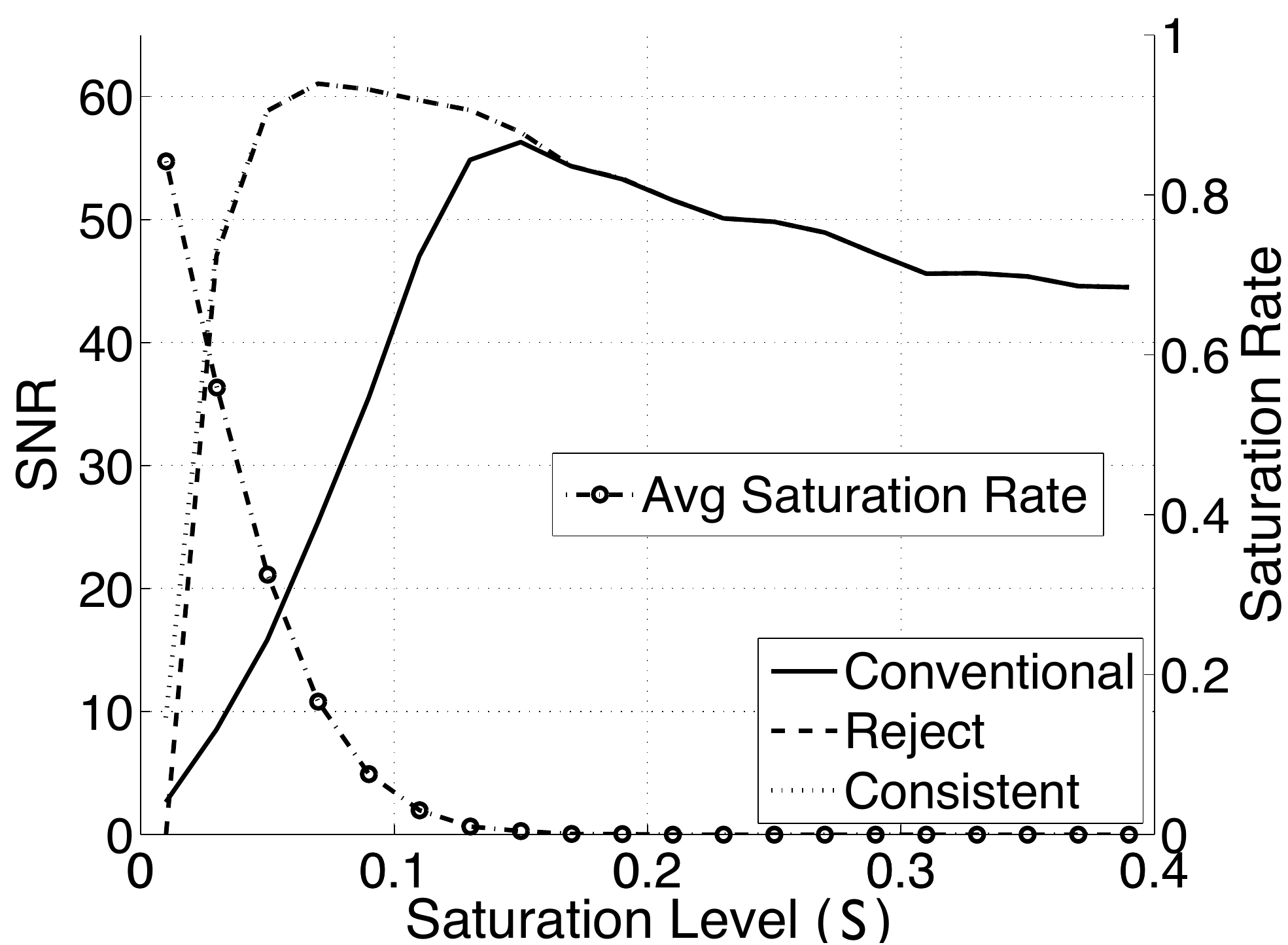}&
  \includegraphics[width=.32\linewidth]{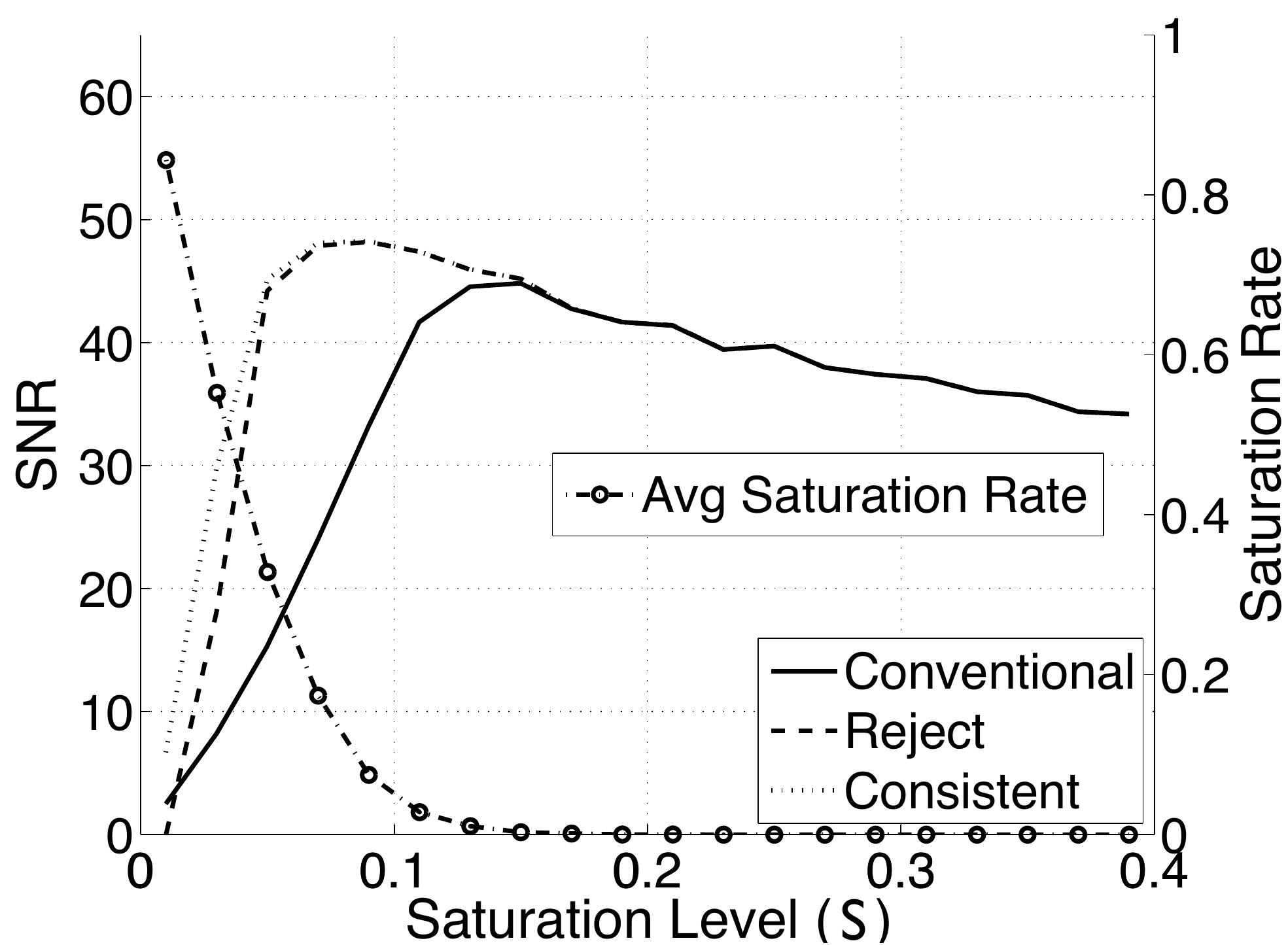}\\
  (a) $k=20$ &(b) $x\in w\ell_{0.4}$&(c) $x\in w\ell_{0.8}$\\
  \end{tabular}
   \caption{Saturation performance using $\ell_1$ minimization for (a)
     exactly $k=20$-sparse signals and compressible signals in weak
     $\ell_p$ for (b) $p=0.4$ and (c) $p=0.8$, with $n=1024$, $m=384$,
     and $\qB=4$.  The reconstruction SNR as a function of the
     saturation level is measured on the left y-axis assuming (solid line)
     conventional reconstruction, \ie ignoring saturation, (dotted
     line) enforcing saturation consistency, and (dashed line)
     rejecting saturated measurements. The dashed-circled line,
     measured on the right y-axis, plots the average saturation rate
     given the saturation level.}
   \label{fig:saturation}
\end{figure}

Fig.~\ref{fig:saturation} demonstrates the effect of each approach. As
demonstrated in the plots, rejecting saturated measurements or
treating them as consistency constraints significantly outperforms
just ignoring saturation. Furthermore, if saturation is properly taken
into account, a distortion optimal finite-range scalar quantizer
should be designed with significant saturation rate, often more than
20\%. While the figures suggest that saturation rejection and
saturation consistency have very similar performance, careful
examination demonstrates, as expected, that consistency provides more
robustness in a larger range of saturation rates and conditions. A
more careful study and detailed discussion can be found
in~\cite{LasBouDav::2009::Demcracy-in-action}. Furthermore, further
gains in the bit-rate can be achieved by coding for the location of
the saturated measurements and transmitting those
separately~\cite{kostina2011value}.
\subsection{1-Bit Compressive Sensing}
\label{sec:1bitCS}

The simplest scalar quantizer design to implement in hardware is a
1-bit quantizer, which only computes the sign of its input. Its
simplicity makes it quite appealing for compressive sensing systems.

The sensing model of 1-bit CS, first introduced in~\cite{BB_CISS08},
is very similar to the standard scalar quantization model
\begin{equation}
  \label{eq:1bit-CS}
  \qq = \sign(A x),
\end{equation}
where $\sign(x_i)$ is a scalar function applied element-wise to its
input and equals $1$ if $x_i\ge 0$ and $-1$ otherwise.

One of the challenges of this model is that it is invariant under changes of the
signal amplitude since $\sign(cx)=\sign(x)$ for any positive $c$. For
that reason, enforcing consistency is not straightforward. A signal
can be scaled arbitrarily and still be consistent with the
measurements. Thus, a magnitude constraint is typically necessary. Of
course, the signal can only be recovered within a positive scalar
factor.

Similarly to multi-bit scalar quantization models, the literature in
this area focuses on deriving lower bounds for the achievable
performance, reconstruction guarantees, as well as practical
algorithms to invert this problem.

\begin{figure}[!Htb]
  \centering
  \subfigure[Orthants in measurement space]{
    \includegraphics[width=.43\linewidth]{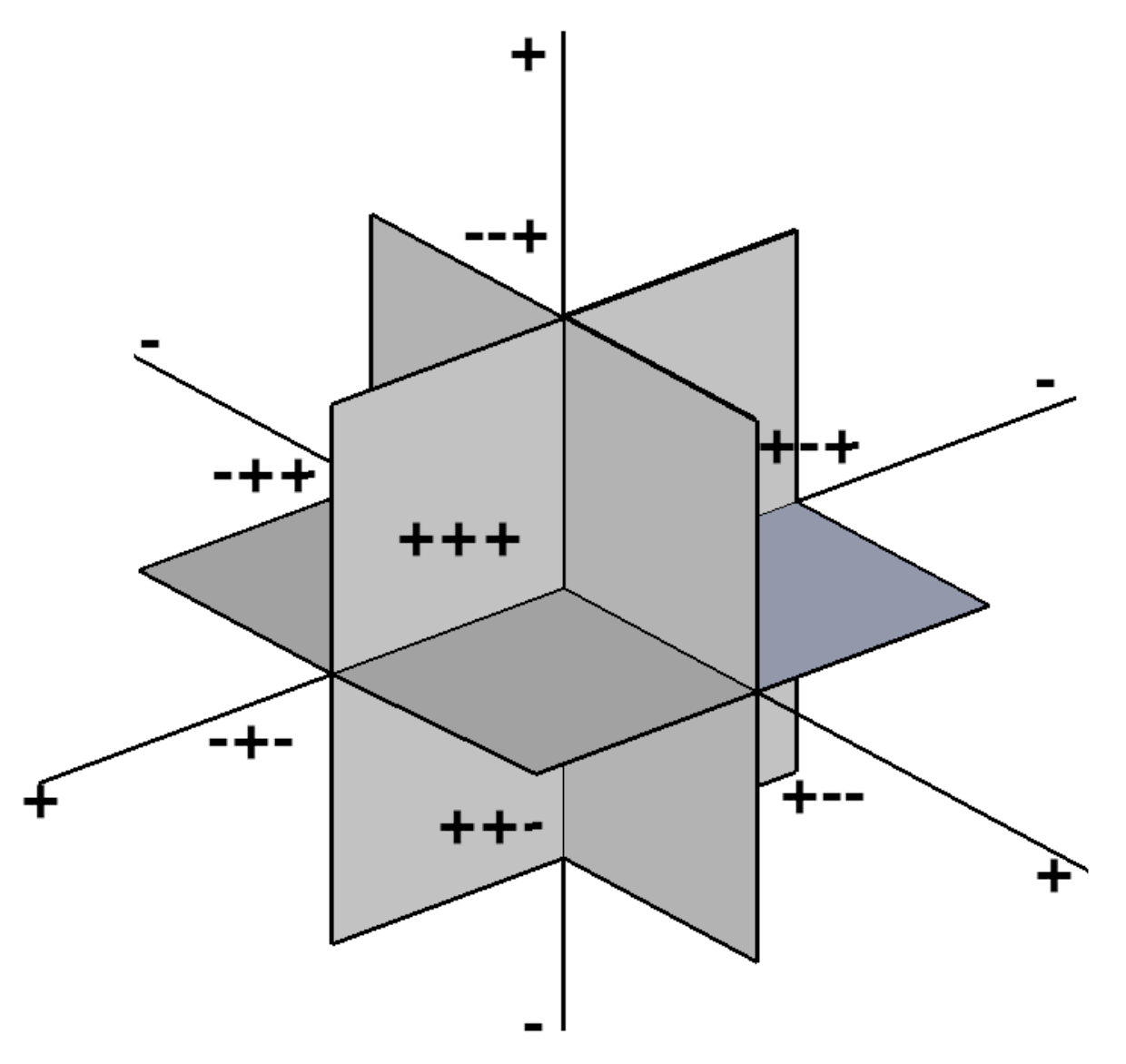}
  }
  \subfigure[Orthant intersection]{
    \includegraphics[width=.43\linewidth]{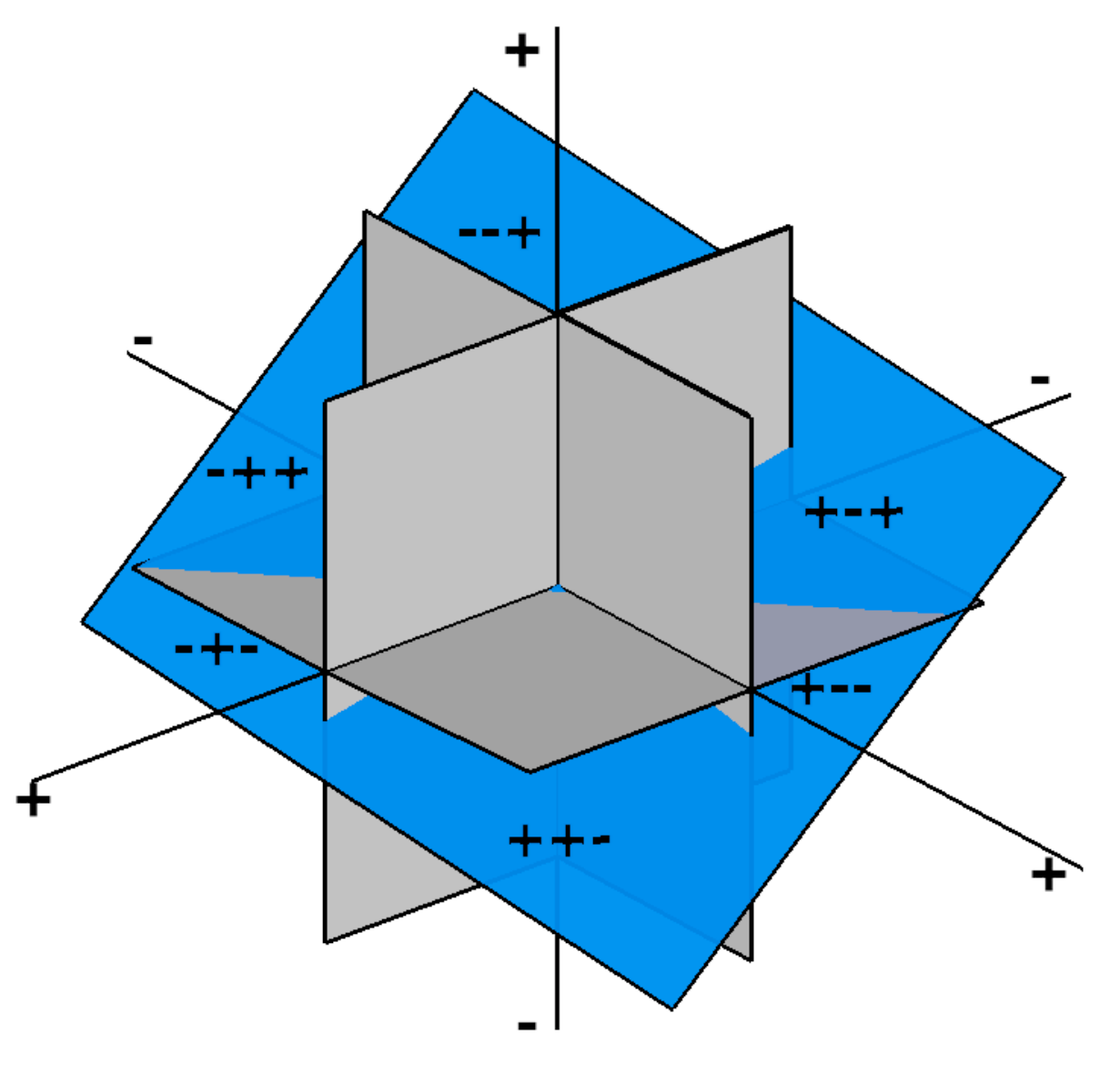}
  }
  \subfigure[Signal space consistency]{
    \includegraphics[width=.5\linewidth]{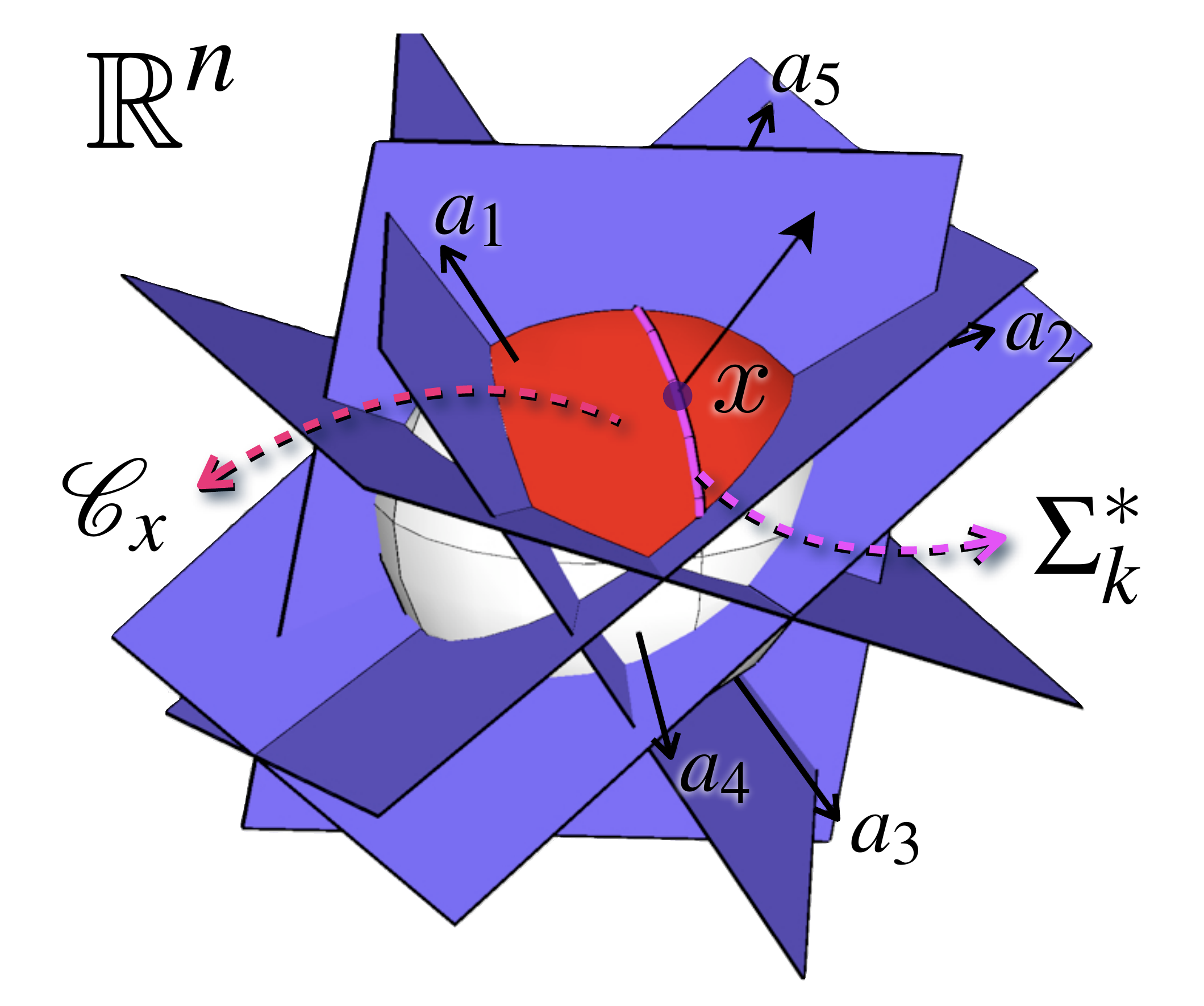}
  }
   \caption{Behavior of 1-bit measurements. (a) The measurement space,
     $\mathbb{R}^m$ is separated to high-dimensional orthant,
     according to the sign of each orthant. (b) Signals in a
     $k$-dimensional space ($k<m$) will only map to a $k$ dimensional
     subspace of $\mathbb{R}^m$ and intersect only a few orthants of
     the measurement space. (c) The same behavior in the signal
     space. Each measurement vector defines its orthogonal
     hyperplane. The measurement sign identifies which side of the
     hyperplane the signal lies on; all signals in the shaded region
     have consistent measurements. Newer measurements provide less and
     less information; the chance of intersecting the consistency
     region decreases.}
  \label{fig:1bit-cs-intuition}
\end{figure}

\subsubsection{Theoretical Performance Bounds}
\label{bounds}

A lower bound on the achievable performance, can be derived using a
similar analysis as in Sec.~\ref{sec:perf-bounds-sparse}. The main
difference is that the quantization cells are now orthants in the
$m$-dimensional space, shown in Fig.~\ref{fig:1bit-cs-intuition}(a),
corresponding to each measured sign pattern. Each subspace of the ${n
  \choose k}$ possible ones intersects very few of those orthants, as
shown in Fig.~\ref{fig:1bit-cs-intuition}(b), \ie uses very few
quantization points. In total, at most $I \le 2^k {n \choose k}{m
  \choose k}$ quantization cells are intersected by the union of all
subspaces~\cite{jacques2013robust}.

Since the signal amplitude cannot be recovered, the lower bound is
derived on $k$-dimensional spheres and coverings using spherical caps
instead of balls. The derivation ensures that the spherical caps have
radius sufficiently large to cover the ${n \choose k}$
spheres. Despite the similarity to the argument in
Sec.~\ref{sec:perf-bounds-sparse}, this case requires a little bit
more care in the derivation; details can be found
in~\cite{jacques2013robust}. Still, the result is very similar in
nature. Defining $\Sigma_{k}^{*}:=\{x\in \Sigma_k, \|x\|_2 =1 \},$ we have:
\begin{theorem}[\cite{jacques2013robust}]
  \label{res:lower_bound}
  Given $x \in \Sigma_{k}^{*}$, any 
  estimation $\hat x \in \Sigma^*_k$ of $x$ obtained from $q = \sign(A x)$ has
  a reconstruction error of at least 
  $$
    \|\hat x - x\|\ \gtrsim\ \frac{k}{m + k^{3/2}},
  $$
  which is on the order of $\tfrac{k}{m}$ as $m$ increases.
\end{theorem}

If the sensing matrix $A$ is Gaussian, \ie if $a_{ij} \sim_{\rm iid}
\cl N(0,1)$, any $k$-sparse signal that has consistent measurements
will not be very far from the signal producing the measurements,
assuming a sufficient number of them. This guarantee approaches the
lower bound of Theorem~\ref{res:lower_bound} within a logarithmic factor.

\begin{theorem}[\cite{jacques2013robust}]
\label{thm:hash}
Fix $0\leq \eta\leq 1$ and $\epsilon_{o}>0$.  If the number of
measurements is
\begin{equation}
\label{eq:gaussian-number-meas-consist}
m \geq \tfrac{2}{\epsilon_{o}}\,\big(2k\,\log(n) + 4k
\log(\tfrac{17}{\epsilon_{o}}) + \log
\tinv{\eta}\big),
\end{equation}
then for all $x, x'\in \Sigma^*_k$  we have that
\begin{equation}
\label{eq:radius-implication}
\|x - x'\|_{2} > \epsilon_{o}\ \Rightarrow\ \sign( A x) \neq \sign (A
x'),  
\end{equation}
with probability higher than $1-\eta$. Equivalently, if $m$ and $k$ are given, solving for $\epsilon_0$ above leads to
\begin{equation}
\label{eq:radius-implication-bound}
\|x - x'\|_{2} \lesssim \tfrac{k}{m} \log \tfrac{mn}{k},
\end{equation} 
with the same probability. 
\end{theorem}

Fig.~\ref{fig:1bit-cs-intuition}(c) provides further intuition on
these bounds by illustrating how 1-bit measurements operate in the
signal space. Specifically, each measurement corresponds to a
hyperplane in the signal space, orthogonal to the measurement
vector. The measurement sign determines on which side of the
hyperplane the signal lies. Furthermore, the signal is sparse, \ie
lies in $\Sigma_k^*$. A consistent sparse reconstruction algorithm can
produce any sparse signal in the indicated shaded region.

A new measurement provides new information about the signal only if
the corresponding hyperplane intersects the region of consistent
signals and, therefore, shrinks it. However, as more measurements are
obtained and the consistency region shrinks, newer measurements have
lower probability of intersecting that region and providing new
information, leading to the $1/m$ decay of the error.

\
Consistency can be quantified using the normalized hamming distance
between measurements
$$
d_H(\qq,\qq')=\frac{1}{m}\sum_i\qq_i\oplus\qq_i',
$$ where $\oplus$ denotes the exlusive-OR operator.  It is, thus,
possible to show that if $x$ and $x'$ above differ by no more than $s$
bits in their 1-bit measurements, \ie if $d_H(\sign(Ax),\sign(Ax'))
\leq s/m$, then, with $m \gtrsim \tfrac{1}{\epsilon_0}\,k
\log \max(m,n)$ and with high probability \cite{QIHT},
$$
\|x - x'\|_{2} \leq \tfrac{k+s}{k} \epsilon_{o}.
$$

A bound similar to \eqref{eq:radius-implication-bound} exists for sign
measurements of non-sparse signals in the context of quantization
using frame permutations~\cite{vivekQuantFrame}. In particular,
reconstruction from sign measurements of signals exhibits (almost
surely) an asymptotic error decay rate arbitrarily close to
$O(1/m)$. However, in contrast to Theorem~\ref{thm:hash} this result
holds only for a fixed signal and not uniformly for all signals of
interest.

Note that these results focus on matrices generated using the normal
distribution. It has been shown that matrices generated from certain
distributions do not perform well in this setting, even though they
can be used in standard compressive
sensing~\cite{plan2011dimension}. For instance, consider a random
Bernoulli matrix $A$ such that $a_{ij} = 1$ or $-1$ with equal
probability. In this case, the two distinct sparse vectors
$(1,0,\cdots,0)^T$ and $(1,\lambda,0,\cdots,0)^T$ with $0\leq \lambda
<1$ are $\lambda$ apart and they generate the same quantization vector
$q=\sign(A_1)$, where $A_1$ is the first column of $A$. It is not
possible, therefore, to distinguish those two vectors from their 1-bit
observations by increasing $m$ and guarantee that the reconstruction
error will decay as measurements increase. This counterexample,
however, is exceptional in the sense that such failures can only
happen if the signal can have a very large entry. Under mild flatness
assumptions on the $\ell_\infty$-norm of the signal, arbitrary
subgaussian measurements can be utilized~\cite{ai2014one}.

These results establish lower and upper bounds on distances between
two sparse signals that have (almost) consistent 1-bit
measurements. It is also possible to provide an embedding guarantee
similar to the RIP~\cite{candes2006ssr}. Since the measurement does
not preserve the signal magnitude, we should not expect distances of
signals to be preserved. However, the measurements do preserve angles
between signals.
Defining $
d_S(u,v) = \tinv{\pi} \arccos(u^T v),\quad u,v\in S^{n-1},$ we have: 

\begin{theorem}[Binary $\epsilon$-Stable Embedding (B$\epsilon$SE) \cite{jacques2013robust}]
\label{thm:1-bit-stable-embed}
Let $A \in \bb R^{m\times n} $ be a random Gaussian matrix such that
$a_{ij} \sim_{\rm iid} \cl N(0,1)$.  Fix $0\leq \eta\leq 1$ and $\epsilon>0$.  If the
number of measurements satisfies
\begin{equation}
\label{eq:numM}
m\ \geq\ \tfrac{2}{\epsilon^2}\big(k\,\log(n) +
2k\,\log(\tfrac{35}{\epsilon}) + \log(\tfrac{2}{\eta})
\big),
\end{equation}
then with probability exceeding $1-\eta$
\begin{equation}
  \label{eq:bse-def}
  d_S(x,x') - \epsilon \leq d_H(\sign(Ax), \sign(Ax')) \leq d_S(x,x')
  + \epsilon,
\end{equation}
for all $x,x' \in \Sigma^*_k$.
\end{theorem} 
 
In other words, up to an additive distortion that decays as
$\epsilon\lesssim (\tfrac{k}{m} \log \tfrac{mn}{k})^{1/2}$, the Hamming
distance between $\sign(Ax)$ and $\sign(Ax')$ tends to concentrate
around the angular distance between $x$ and $x'$. Notice that, in
contrast to the RIP, a vanishing distance between the quantized
measurements of two signals does not imply they are equal, \ie we observe a (restricted) \emph{quasi-isometry}
between $\Sigma^*_k$ and $\sign(A \Sigma^*_k)$ instead of the common
RIP \cite{jacques2013quantized}. This comes from the additive nature of the distortion in
\eqref{eq:bse-def} and is a direct effect of the inherent ambiguity
due to quantization.
 
This embedding result has been extended to signals belonging to
convex sets $\cl K\subset \bb R^n$ provided that their Gaussian mean
width
\begin{equation}
  \label{eq:gaussian-width}
  w(\cl K) = \bb E\, \sup\{ u^T g: u\,\in\,\cl K\! -\! \cl K\},\quad g \sim
  \cl N(0,{\rm I}_{n\times n}),
\end{equation}
with \mbox{$\cl K\!\!-\!\cl K :=\{v - v': v,v' \in \cl K\}$}, can be computed
\cite{plan2013one,plan2012robust,plan2011dimension}. In particular, if
$$
m \geq C \epsilon^{-6} w^2(\cl K)
$$ for some constant $C>0$, then \eqref{eq:bse-def} holds with high
probability for any $x,x' \in \cl K \cap S^{n-1}$. In particular, for
$$
\cl K = K_{n,k} := \{u \in \bb R^n: \|u\|_1 \leq k^{1/2}, \|u\|_2 \leq 1\},
$$ since $w^2(K_{n,k}) = O(k \log n/k)$ \cite{plan2012robust}, an
embedding exists between the set of compressible vectors modeled by $K_{n,k}$
and  $\{-1,+1\}^m$ provided that $m \geq C \epsilon^{-6}\,k \log n/k$.

Note that generalizations of these embeddings to non-linear functions
other than the $\sign$ operator, or to stochastic processes whose
expectation is characterizable by such functions, are also possible
\cite{plan2012robust}. 

\subsubsection{Reconstruction from 1-Bit Measurements}
\label{sec:reconstr-greedy-conv}
The original efforts in reconstructing from 1-bit measurements
enforced $\|x\|_2=1$ as a reconstruction constraint, formulating the
non-convex $\ell_1$ minimization problem
\begin{align}
  \hat{x}=\arg\min_x\|x\|_1,~\mathrm{s.t.}~\qq=\sign(Ax),~\|x\|_2=1. \label{eq:1bmin}
\end{align}
Even though the problem is not convex, a number of algorithms have
been shown experimentally to converge to the
solution~\cite{BB_CISS08,LasWenYin::2010::Trust-but-verify:}. More
recently, a number of greedy algorithmic alternatives have also been
proposed~\cite{Bou::2009::Greedy-sparse,bahmani2013robust,jacques2013robust}.

Most of these algorithms attempt to enforce consistency by introducing
a one-sided penalty for sign violations
\begin{align}
  J(Az,\qq) = \| (\qq \circ A z)_- \|_q,
  \label{eq:one-sided}
\end{align}
where $\circ$ is the element-wise product between vectors, $(y_i)_- =
y_i$ if $y_i$ is negative and 0 otherwise, also applied element-wise,
and the $\ell_q$ norm is typically the $\ell_1$ or the $\ell_2$
norm. Typically, a descent step is performed using the gradient
of~\eqref{eq:one-sided}, followed by a support identification and
sparsity enforcement step. Often, care is taken in selecting the
descent step, especially considering the signal is on the unit
$\ell_2$ sphere~\cite{LasWenYin::2010::Trust-but-verify:}. Assuming a
certain noise level, a maximum likelihood formulation can also be used
to remove the norm constraint~\cite{bahmani2013robust}.

For example, the Binary IHT (BIHT), a variation of the popular Iterative Hard
Thresholding (IHT)~\cite{BluDav::2008::Iterative-hard}, uses the
one-sided $\ell_1$ norm in~\eqref{eq:one-sided} and follows its
subgradient $\tinv{2}A^T(q-\sign(Az))$. The algorithm is defined by
the iteration
\begin{equation}
  \label{eq:biht-def}
  z^{n+1} = \cl H_k\big(z^{n} + \tinv{2} A^T(y - \sign(A z^{n}))\big), \quad z^{0} = 0,
\end{equation}
where $\cl H_k(\cdot)$ is a hard threshold, keeping the largest $k$
coefficients of its input and setting the remaining ones to zero.

The BIHT does not have convergence or reconstruction guarantees to a
consistent output. Still, as shown in Fig.~\ref{fig:AngVsHamm}, it
works surprisingly well compared to other greedy approaches. Moreover,
variations exist to make it more robust to potential binary errors in
the knowledge of~$q$~\cite{jacques2013robust} or to extend it to
multi-bit scalar quantization \cite{QIHT}.

\begin{figure}[t]
   \centering
   \subfigure[\label{fig:AngVsHamm-01}m/n = 0.1]{\includegraphics[width=.305\textwidth]{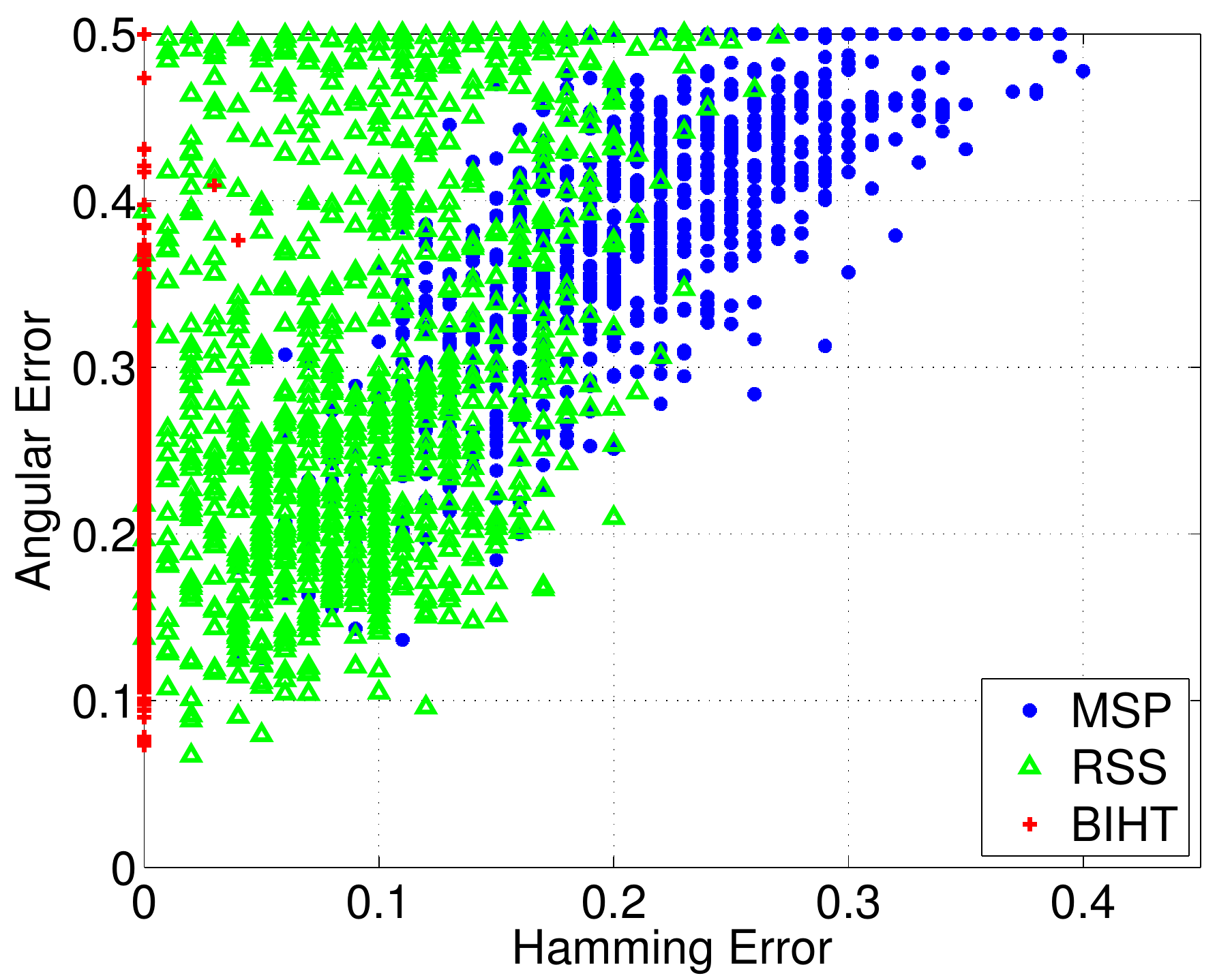}}
   \subfigure[\label{fig:AngVsHamm-07}m/n = 0.7]{\includegraphics[width=.32\textwidth]{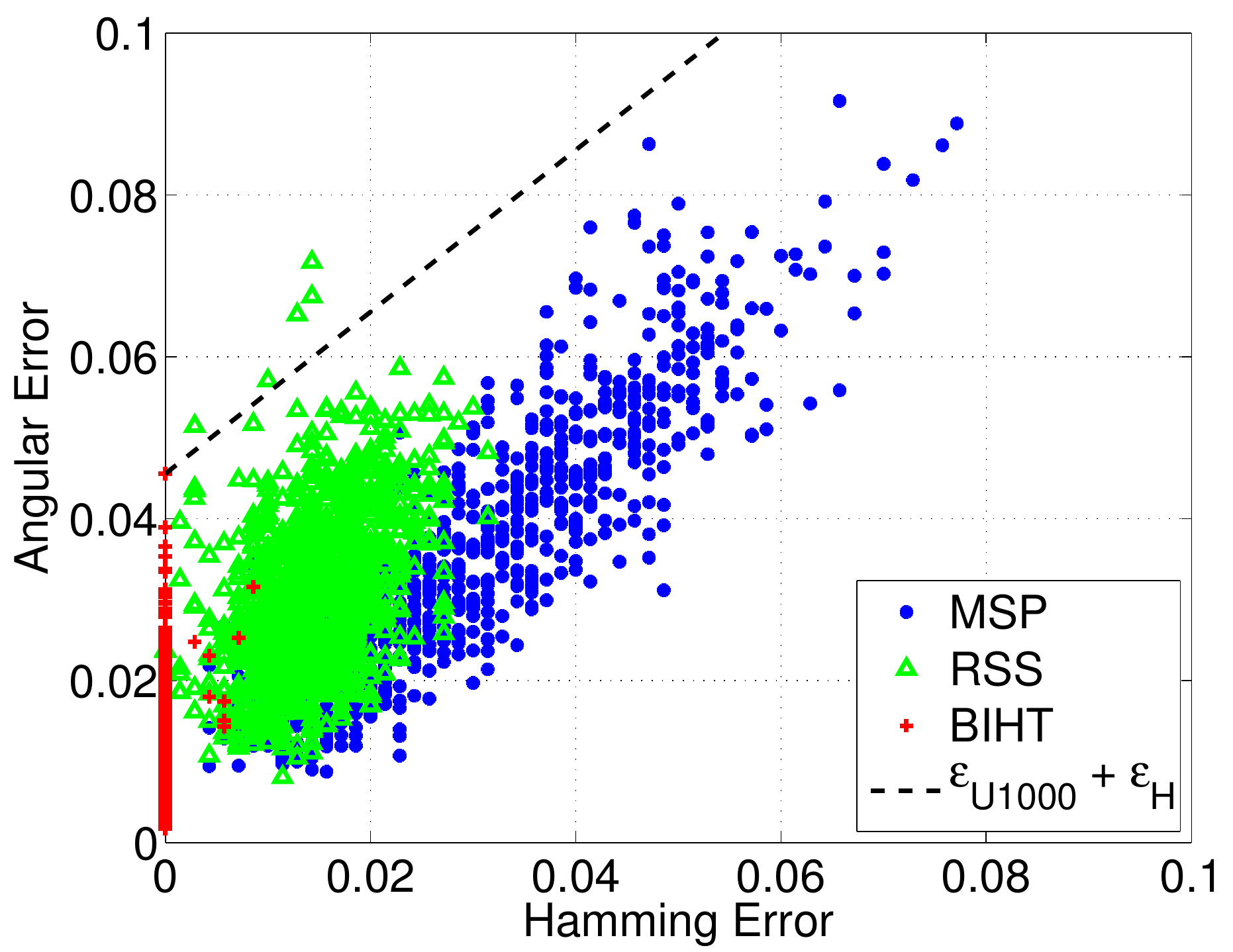}}
   \subfigure[\label{fig:AngVsHamm-15}m/n =1.5]{\includegraphics[width=.32\textwidth]{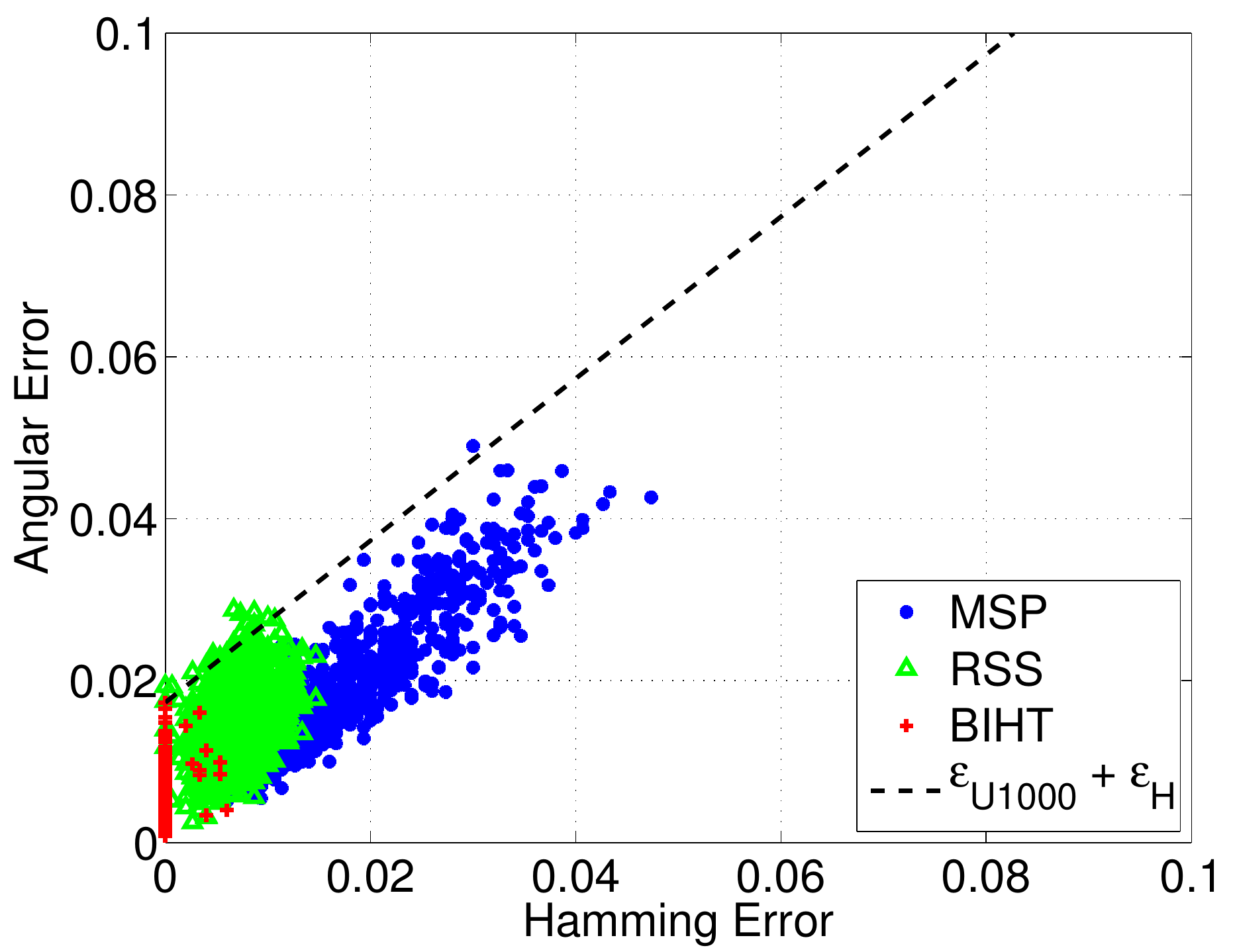}}
   \caption{Angular reconstruction error $\epsilon_S = d_S(x,\hat x)$
     vs.\ consistency error $\epsilon_H = d_H(\sign(A\hat x), q)$ for different
     greedy reconstructions (MSP, RSS and BITH).  BIHT returns a
     consistent solution in most trials.  When $A$ is a B$\epsilon$SE,
      \eqref{eq:bse-def} predicts that the angular
     error $\epsilon_S$ is bounded by the hamming error $\epsilon_{H}$
     (and conversely) in addition
     to an offset $\epsilon$. This phenomenon is confirmed by an
     experimental linear trend (in dashed) between the two errors that improves
     when $m/n$ increases \cite{jacques2013robust}. \label{fig:AngVsHamm}}
\end{figure}

The first iteration of BIHT is a simple truncated back-projection,
$\hat x_0 = \cl H_k(A^T q)$ whose distance to $x$ is known to decay
asymptotically as
$\sqrt{k/m}$ for a Gaussian matrix $A$
\cite{QIHT,bahmani2013robust}. Furthermore, $\hat x_0$ matches the
solution of the (feasible) problem
$$ \arg\max_z q^T Az\ {\rm s.t.}\ z\in\Sigma_k^*,
$$
where maximizing $q^T Az$ also promotes the 1-bit consistency of $z$
with $q$. 

This optimization can be generalized to any convex sets $\cl K \subset
\bb R^n$ where $x$ can lie, such as the set $\cl K = K_{n,k}$ of
compressible signals with Gaussian width $w(K_{n,k}) \asymp k \log
n/k$~\cite{plan2012robust}. If $m \geq C \epsilon^{-2} w(\cl K)^2$
for some $C>0$, and a fixed $x$ is sensed using~\eqref{eq:1bit-CS}
with a Gaussian sensing matrix $A$, then the solution to
$$
\hat x\ =\ \arg\max_z\ q^T Az\ {\rm s.t.}\ z\in\cl K,
$$ satisfies $\|\hat x - x\|^2 = O(\epsilon)$ with high probability.
Interestingly, under certain conditions, this holds also for sensing
models other than \eqref{eq:1bit-CS}, where the sign operator is
replaced, for instance, by the logistic function
\cite{plan2012robust}.

What makes it difficult to provide reconstruction error estimates for algorithms 
motivated by the problem
\eqref{eq:1bmin} is the non-convex constraint $\|x\|_2=1$, whose convex relaxation 
allows for the zero solution and is hence meaningless.
To overcome this obstacle, it has been proposed in~\cite{plan2013one,plan2012robust} to
impose a norm constraint to prevent trivial solutions on the measurements rather than the signal.
This results in a different problem, which allows for a meaningful convex
relaxation. Namely, since $\qq=sign(Ax)$, it
follows that at the solution $\qq^T(Ax)=\|Ax\|_1$. Thus, by
constraining this norm, the following convex problem can be
formulated:
\begin{align}
  \hat{x}=\arg\min_x\|x\|_1,~\mathrm{s.t.}~\qq=\sign(Ax),~\qq^TAx=1
  \label{eq:1bit_convex}
\end{align}
As shown in \cite{plan2013one}, the problem in \eqref{eq:1bit_convex}
does allow for reconstruction guarantees: If $m \sim \epsilon^{-5} k
\log (n/k)$, the solution $\hat x$ recovered from quantized Gaussian
measurements of a sparse signal $x$ is such that $d_S(x,\hat x) \leq
\epsilon$ with high probability. This holds uniformly for all signals
$x \in \bb R^n$. Under flatness assumptions on the signal, recovery
guarantees can also be proved for arbitrary subgaussian measurements
\cite{ai2014one}.
\subsection{Noise, Quantization and Tradeoffs}
\label{sec:extensions}
   
The sections above were focused on noiseless QCS models. These models
only consider the statistical or the geometrical properties of
quantization of CS measurements under high or low resolution
modes. However, any signal acquisition system is subject to noise
corruption before quantization, either on the measurement process or
on the signal itself. Such noise can be incorporated in a more general
model
\begin{equation}
  \label{eq:qcs-pcm-noisy}
  q = \qQop(A (x + \noise_x) + \noise_{\rm s}),  
\end{equation}
where $\noise_x\in \bb R^n$ and $\noise_{\rm s} \in \bb R^m$ corrupt
the signal and the sensing, respectively, before
quantization. Examining the impact of such noise in signal recovery
leads to new interesting questions.

In \cite{zymnis2010compressed} two efficient reconstruction methods are
developed for sparse or compressible signals sensed according
\eqref{eq:qcs-pcm-noisy} under sensing noise only, \ie $\noise_{x} = 0$. The
two approaches are mainly numerical: one relies on a maximum likelihood
formulation, built on the quantization model and on a known Gaussian noise
distribution, the other follows a least square principle. The two resulting methods are both
regularized by an $\ell_1$-norm accounting for sparse signal prior. A
provably convergent procedure inherited from a fixed point
continuation method is used for reconstructing the signal in the two
possible frameworks. With their approach, the combined effects of
noise and coarse quantization can be jointly handled. Reasonable
reconstruction results are achieved even using 1 or 2 bits per
measurement.

The case $\noise_x \neq 0$, $\noise_{\rm s} = 0$ boils down to an interaction of the
well-understood phenomenon of \emph{noise folding} in CS
\cite{davenport2012pros} and quantization \cite{laska2012regime}. Noise-folding in unquantized CS
says that under a weak assumption of orthogonality between the rows of $A$,
the variance of the component $A \noise_x$ undergoes a multiplication
by $n/m$ compared to the variance $\sigma_{\noise}^2$ of
$\noise_x$. This impacts directly the reconstruction error of
signals. The corresponding MSE is then $n/m$ times higher than the
noise power, or equivalently, the SNR looses 3 dB each time $m$ is
divided by 2 \cite{davenport2012pros}.

An extension of this result to noisy QCS has been provided
in~\cite{laska2012regime}, assuming the sensing matrix $A$ is RIP of
order $k$ and constant $\delta$. In this case, if $\noise_x$ is
standard normally distributed and if the quantizer has resolution $B$,
then, under a random signal model where the signal support $T$ is
chosen uniformly at random in $\{1,\cdots, n\}$ and the amplitudes of
the non-zero coefficients are standard normally distributed,
\begin{equation}
  \label{eq:regime-change-tradeoff}
(1-\delta)\bb E \|x - \hat x\|^2 = 2^{-2B+1} \tfrac{k}{m} \bb E \|x\|^2 
+ 2\big(2^{-2B} 
+  1 \big) \tfrac{n}{m}\,\bb E \|\noise_x|_T\|^2 + km \kappa, 
\end{equation}
where $\hat x = (A_T^\dagger q)_T$ is the oracle-assisted
reconstruction of $x$ knowing the
support $T$ of $x$ for each of its realization, and 
$$
\kappa =
\max_{i\neq j} |\bb E \qQop( a_i^T(x+\noise_x)) \qQop(
a_j^T(x+\noise_x))|,
$$
measures the worst correlation between distinct
quantized measurements.

In \eqref{eq:regime-change-tradeoff}, the first term
accounts for the quantization error of the signal itself, while the
second term represents both the error due to folded signal noise as
well as the quantization of that noise. Finally, the third term
reflects a distortion due to correlation between quantized
measurement. It is expected to be negligible in CS scenarios,
specially when $B$ increases or if a \emph{dithering} is added to $\qQop$ \cite{gray1998quantization}.

\begin{figure*}[!t] 
   \centering
   \begin{tabular}{cc}
     \includegraphics[width=.4\textwidth]{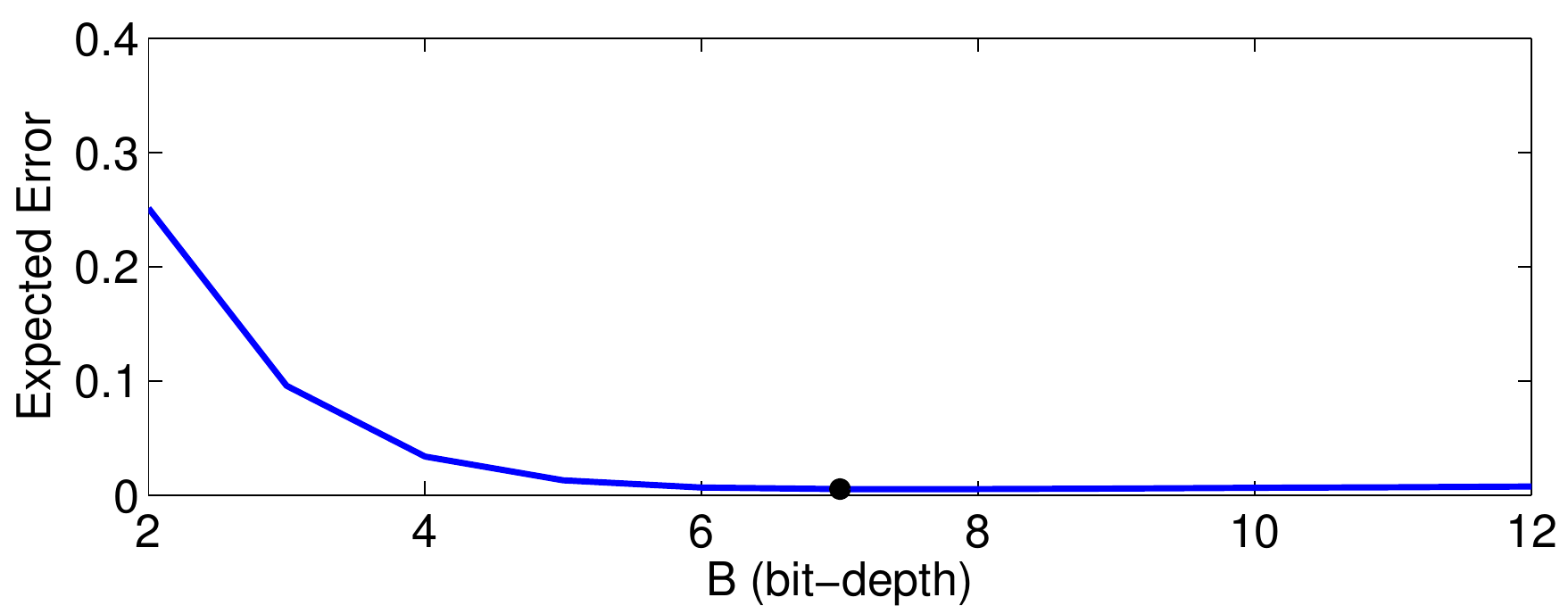}& 
     \includegraphics[width=.4\textwidth]{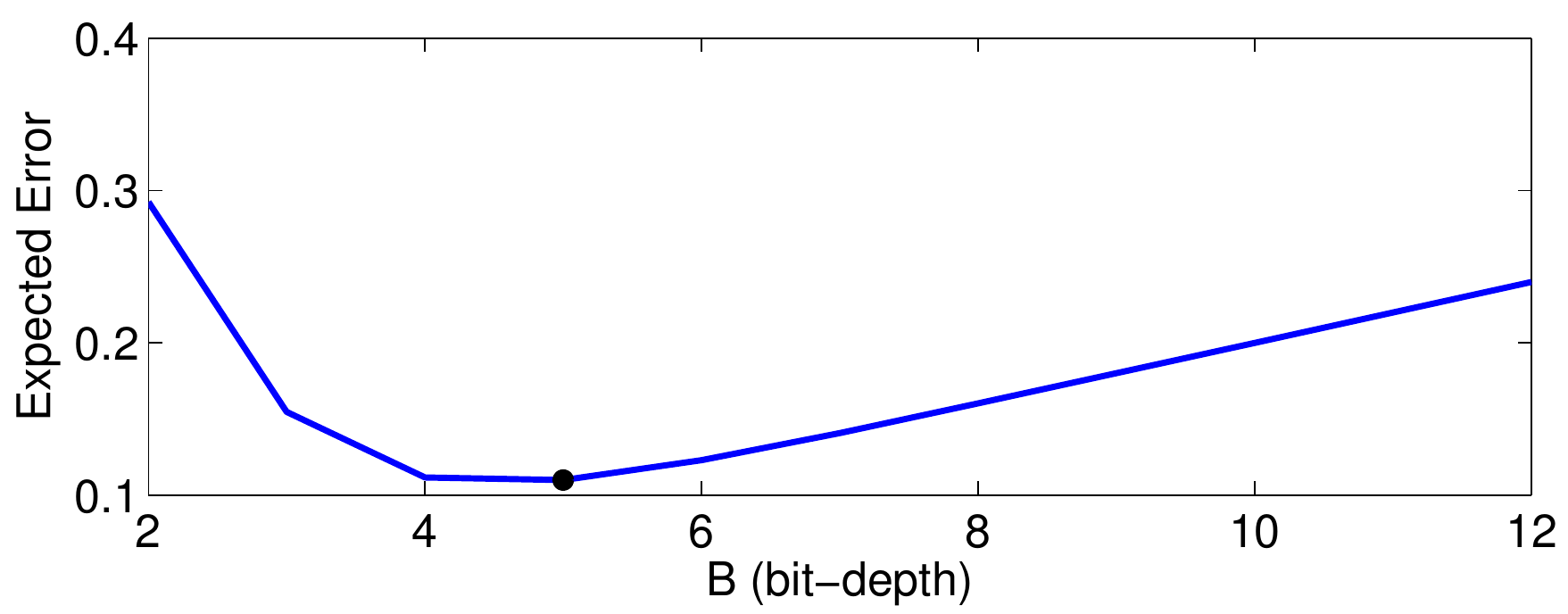}\\
       (a) \small{$\mathrm{ISNR} = 35$dB, optimal bit-depth $= 7$}&(b) \small{$\mathrm{ISNR} = 20$dB, optimal bit-depth $= 5$}\\
          \includegraphics[width=.4\textwidth]{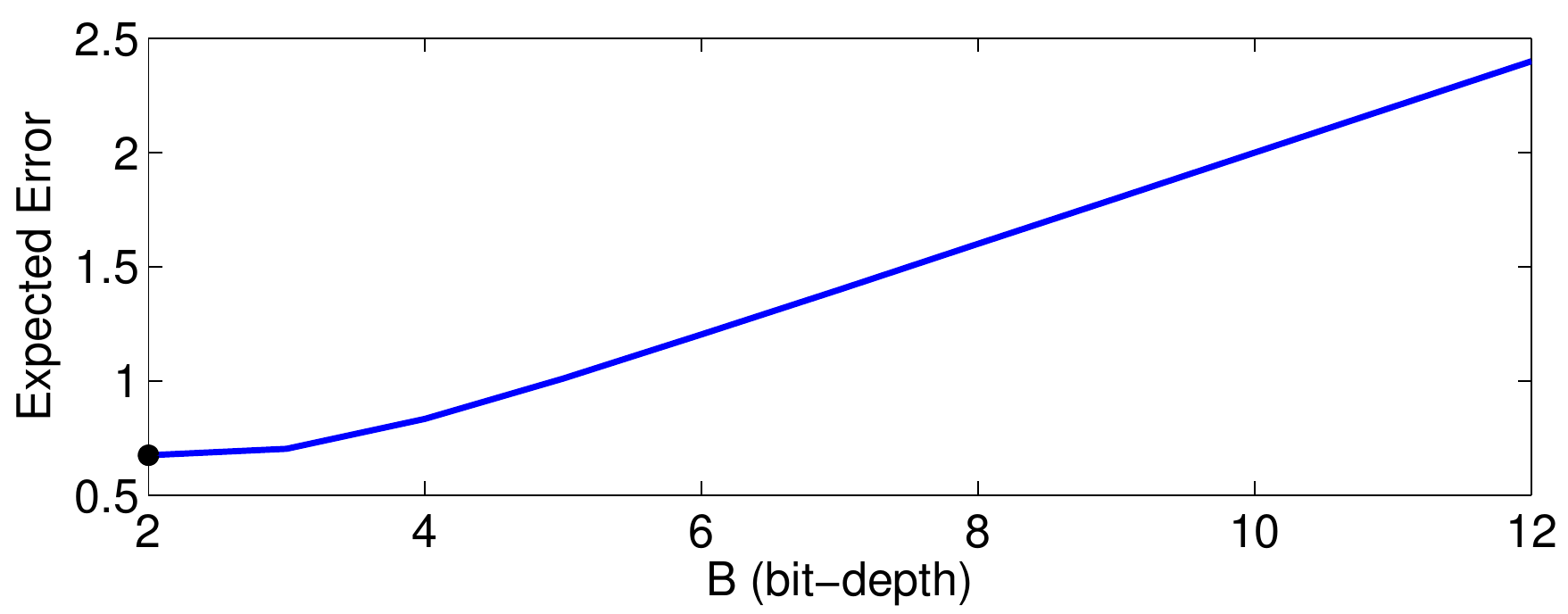}&
   \includegraphics[width=.4\textwidth]{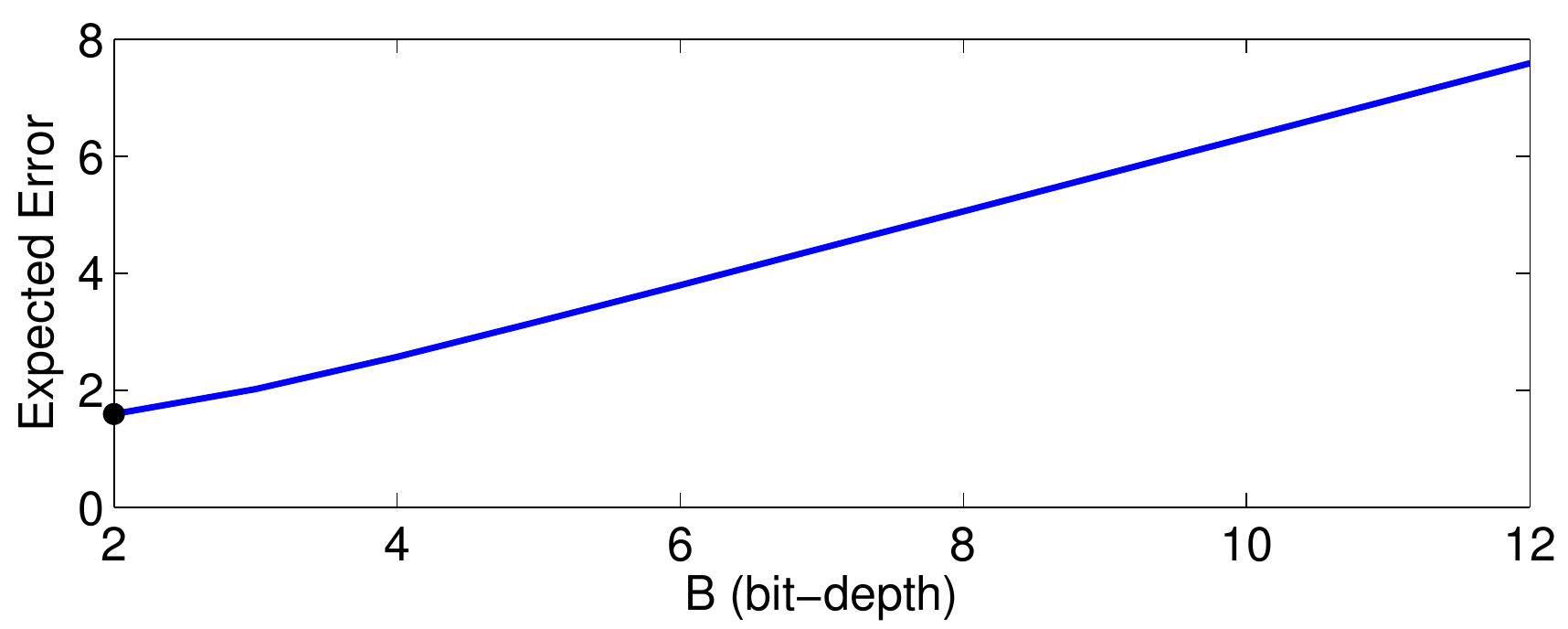}\\
	(c) \small{$\mathrm{ISNR} = 10$dB, optimal bit-depth $= 2$}&(d)\small{ $\mathrm{ISNR} = 5$dB, optimal bit-depth $= 2$}\\
   \end{tabular}
   \caption{Upper bound on the oracle-assisted reconstruction error as
     a function of bit-depth $B$ and ISNR at
constant rate $R = mB$ \cite{laska2012regime}.  
The black dots denote the minimum point on each curve. 
   }
   \label{fig:bound}
\end{figure*}

Numerical study of \eqref{eq:regime-change-tradeoff} shows that, at
constant rate $R = mB$, a tradeoff can be expected between a
\emph{measurement compression} (MC) regime, where $m$ is small (but
still high enough to guarantee $A$ to be RIP) and $B$ is high, and a
\emph{quantization compression} (QC) regime, where $m$ is high
compared to the standard CS setting but $B$ is small. Interestingly,
the optimal bit-depth $B$, minimizing the expected reconstruction
error, depends on the input SNR: $\mathrm{ISNR}=20 \log_{10}
\|x\|/\|\noise_x\|$. This is illustrated in Fig.~\ref{fig:bound} where
the evolution of \eqref{eq:regime-change-tradeoff} (discarding the
effect of the third term) is plotted for four different noise
scenarios. The optimal bit depth decays smoothly with the ISNR,
suggesting that the QC regime is preferable at low ISNR while MC is
clearly better at high ISNR. The general behavior of
Fig.~\ref{fig:bound} is also confirmed on Monte Carlo error estimation
of the oracle-assisted reconstruction defined above
\cite{laska2012regime}.

\section{Sigma-Delta Quantization for Compressive Sensing}\label{sec:SDCS}

As mentioned in the introduction, $\sd$ quantization for compressed
sensing fundamentally builds on corresponding schemes for finite
frames. Thus before presenting an analysis specific to compressed
sensing, we first discuss the finite frame case.

\subsection{$\sd$  Quantization for Frames}\label{subsec:frame}
Let $\Phi \in \mathbb{R}^{n\times N}$  with columns $\{ \phi_j\}_{j=1}^N$ be a frame in the sense of (1.32) and consider the frame expansion $$c=\Phi^T x$$ of a signal $x\in {\mathbb R}^n$. The goal is now to quantize  $c$ as a whole such that the quantized representation $q$ allows for approximate recovery of $x$. $\sd$ quantization schemes obtain such a $q$ using a recursive procedure, which we will now explain in detail.

At the core of the schemes is a uniform scalar quantizer $\qQop$, which maps a real number to the closest point in a codebook of the form 
\begin{equation}\qQal = \big\{ (\pm j-1/2)\qD, j \in \{ 1, ... ,  L\}\big\}.\label{eq:SD_codebook}\end{equation}

A $\sd$ scheme applies such a quantizer sequentially to the entries of $c$, taking in each quantization step the errors made in $r$ previous steps into account. The complexity parameter $r$ is referred to as the order of the $\sd$ scheme; it quantifies the trade-off between required storage and achievable accuracy.

A first order $\sd$ quantization scheme, the simplest such algorithm, hence retains the error only for one step.   In the following formalization associated with the so-called \emph{greedy} first order $\sd$ scheme, the error parameter appears as the state variable $u_i$; it measures the total accumulated error up to step $i$.
The quantized frame coefficient vector $q\in \qQal^N$ is computed by running the iteration
\begin{align}
q_i &= \qQop (u_{i-1} + c_i) \nonumber \\
u_i &= u_{i-1} + c_i - q_i. \label{eq:1st_order}
\end{align}
As initialization, one typically uses $u_0=0$. In matrix-vector notation, the above recurrence relation reads \begin{equation}D u = c - q.\label{eq:mat_form} \end{equation} Here $D\in \mathbb{R}^{N\times N}$ is the finite difference matrix with entries given in terms of the Kronecker delta by $D_{ij} = \delta_{i,j} -\delta_{i+1,j}$, that is, 
\begin{equation}
D = \begin{pmatrix}1&0  & 0 & \cdots & 0\\-1 &1 &0 & & 0\\ 0& -1 & 1 & & 0\\ \vdots& & \ddots&\ddots &\vdots 
\\ 0& 0 &\cdots &-1 &1
\end{pmatrix}.
\end{equation}

The scheme is explicitly designed such that each $q_j$ partly cancels the error made up to $q_{j-1}$. When the signal is approximated as $\widetilde{\Phi} q$ using a dual frame $\widetilde \Phi \in \mathbb{R}^{n\times N}$  with columns $\{ \widetilde \phi_j\}_{j=1}^N$, this entails that one seeks to compensate an error in the direction of a dual frame vector $\widetilde \phi_{j-1}$ using a distortion in the direction of the next dual frame vector $\widetilde \phi_{j}$. This serves as a motivation to choose a smoothly varying dual frame, \ie with subsequent dual frame vectors close to each other.

Bounding the reconstruction error using \eqref{eq:mat_form} in terms of the operator norm $\|A\|_{2\rightarrow 2}:= \sup_{\|x\|_2\leq 1} \|Ax\|_2$, one obtains \[ \| x- \widetilde{\Phi}q \|_2 =  \|\widetilde{\Phi}(c-q)\|_2 = \|\widetilde{\Phi}D u \|_2 \leq \|\widetilde{\Phi} D \|_{2\rightarrow 2} \|u\|_2. \]  The smoothness intuition is reflected in the fact that the columns of $\widetilde{\Phi} D$ are given by $\widetilde{\phi}_j - \widetilde{\phi}_{j-1}. $  Thus more precisely, finding a smooth dual frame $\widetilde \Phi$ amounts to minimizing $\|\widetilde \Phi D\|_{2\rightarrow 2}$.

If one is willing to store more than one previous value of the state variable, that is, to consider a higher order $\sd$ scheme, it is possible to profit from higher order smoothness of the dual frame. Such a generalization of \eqref{eq:1st_order} is the greedy $r$-th order $\sd$ scheme, which is associated with the recurrence relation
\begin{equation}D^r u = c - q.\label{eq:mat_form_r} \end{equation} 
Here, the iteration to compute the quantized coefficients is explicitly given by 
\begin{align}
q_i &\textstyle = \qQop \big(\sum\limits_{j=1}^{r}(-1)^{j-1} {{r}\choose{j}}u_{i-j} + c_i\big) \nonumber \\
u_i &\textstyle = \sum\limits_{j=1}^{r}(-1)^{j-1} {{r}\choose{j}}u_{i-j} + c_i - q_i. \label{eq:r-th_order}
\end{align}
As before, one initializes $u_i=0, \ i\leq 0$. The reconstruction
error is now bounded by
\begin{equation} \| x- \widetilde{\Phi}q \|_2 =  \|\widetilde{\Phi}(c-q)\|_2 = \|\widetilde{\Phi}D^r u \|_2 \leq \|\widetilde{\Phi} D^r \|_{2\rightarrow 2} \|u\|_2. \label{eq:errbound}
\end{equation}
Examining \eqref{eq:errbound}, it is advantageous to choose a dual frame that minimizes $\|\widetilde{\Phi}D^r \|_{2\rightarrow2}$, and a $\sd$ scheme that yields a state-variable sequence with well bounded $\|u\|_2$. This motivates the following definitions.
\begin{definition}
Let $\Phi \in \mathbb{R}^{n\times N}$ be a frame and $r$ be a positive integer. Then the $r$-th order  \emph{Sobolev dual} of $\Phi$ is given by 
\begin{equation}\widetilde\Phi^{(r)}:= \arg\min \| \widetilde{\Phi}D^r\|_{2 \rightarrow 2} = (D^{-r} \Phi)^\dagger D^{-r} \label{eq:Sobolev},\end{equation}
where the minimum is taken over all dual frames of $\Phi$.
\end{definition}

\begin{definition}
A $\sd$ scheme with a codebook $\qQal$ is \emph{stable} if there exist constants $C_1$ and $C_2$ such that whenever $\|c\|_\infty \leq C_1$ we have $\|u\|_\infty \leq C_2.$ 
\end{definition}

In general, designing and proving the stability of $\sd$ quantization schemes of arbitrary order can be quite difficult if the number of elements in the associated codebook is held fixed. This challenge is especially difficult in the case of 1-bit quantizers and overcoming it is the core of the contributions of \cite{DD03, G03, DGK11}, where stable $\sd$ quantization schemes of arbitrary order are designed. On the other hand, if the number of elements in the codebook \eqref{eq:SD_codebook} is allowed to increase with order, then even the simple greedy $\sd$ schemes \eqref{eq:r-th_order} are stable, as the following proposition shows (see, \eg \cite{BPA07}). 

\begin{proposition}
The greedy $r$-th order $\sd$ scheme \eqref{eq:r-th_order} associated with the $2L$-level scalar quantizer \eqref{eq:SD_codebook} is stable, with $\|u\|_\infty \leq \qD/2$, whenever $\|c\|_\infty \leq  \qD (\qL - 2^{r-1}+ 2^{-1})$.
\end{proposition}
\begin{proof}
The proof is by induction. We begin by rewriting \eqref{eq:r-th_order} in terms of auxiliary state variables  $u_i^{(j)}, j=1,...,r$ and $u_i^{(0)}=c_i -q_i$ as 
\begin{align}
q_i &= \qQop \Big(\sum_{j=1}^{r}u^{(j)}_{i-1} + c_i\Big) \nonumber \\
u_i^{(j)} &= u_{i-1}^{(j)}+u_{i}^{(j-1)}, \quad j=1,...,r \label{eq:r-th_order2}
\end{align}
with $u^{(j)}_0=0$ for $j=1,...,r$. Note that with this notation $u_i^{(r)}=u_i$. Now suppose that $|u^{(j)}_{i-1}| \leq 2^{r-j} \qD/2$ for all $j \in \{1,...,r\}$, then $|\sum_{j=1}^r u^{(j)}_{i-1}| \leq (2^r-1) \qD/2$.
Since by the $\sd$ iterations we have $u^{(j)}_{i} = \sum_{k=1}^j u^{(k)}_{i-1} + c_i - q_i$ we deduce that \[ |u_i^{(r)}| = | \sum_{k=1}^r u^{(k)}_{i-1} + c_i - \qQop(\sum_{k=1}^r u^{(k)}_{i-1} + c_i)|  \leq \qD/2\] 
provided  $\|c\|_\infty \leq \qD (\qL - 2^{r-1}+ 1/2)$. Moreover, by \eqref{eq:r-th_order2}, $|u^{(j)}_{i}| \leq 2^{r-j}\Delta/2$.
 \end{proof}

 Working with stable $r$-th order $\sd$ schemes and frames with
 smoothness properties, and employing the Sobolev dual for
 reconstruction, it was shown in \cite{BLPY10} that the reconstruction
 error satisfies $\| x - \widetilde{\Phi}q\|_2 \leq C_{r,\Phi}
 N^{-r}$, where the constant $C_{r,\Phi}$ depends only on the quantization scheme and the frame. Such results for $\sd$-quantization show that its error
 decay rate breaks the theoretical $\sim 1/N$ lower bound of scalar
 quantization described in the introduction.  Fig.~\ref{fig:quant_cells} helps to illustrate why such a result is
 possible. It shows the quantization cells associated with two bit
 quantization of $\Phi^T x$, where $x$ in the unit ball of
 $\mathbb{R}^2$, using both first order $\sd$ quantization and scalar
 quantization. For vectors belonging to a given cell, the worst case
 error achieved by an optimal decoder is proportional to the diameter
 of the cell.  The figure shows that the cells resulting from $\sd$
 quantization are smaller than those resulting from scalar
 quantization, indicating the potential for a smaller reconstruction
 error. For a detailed overview of $\sd$ quantization of frame
 expansions see, \eg \cite{PSY13}.

\begin{figure}[t]
  \centering
   \begin{tabular}{cc}
  \includegraphics[width=.5\linewidth]{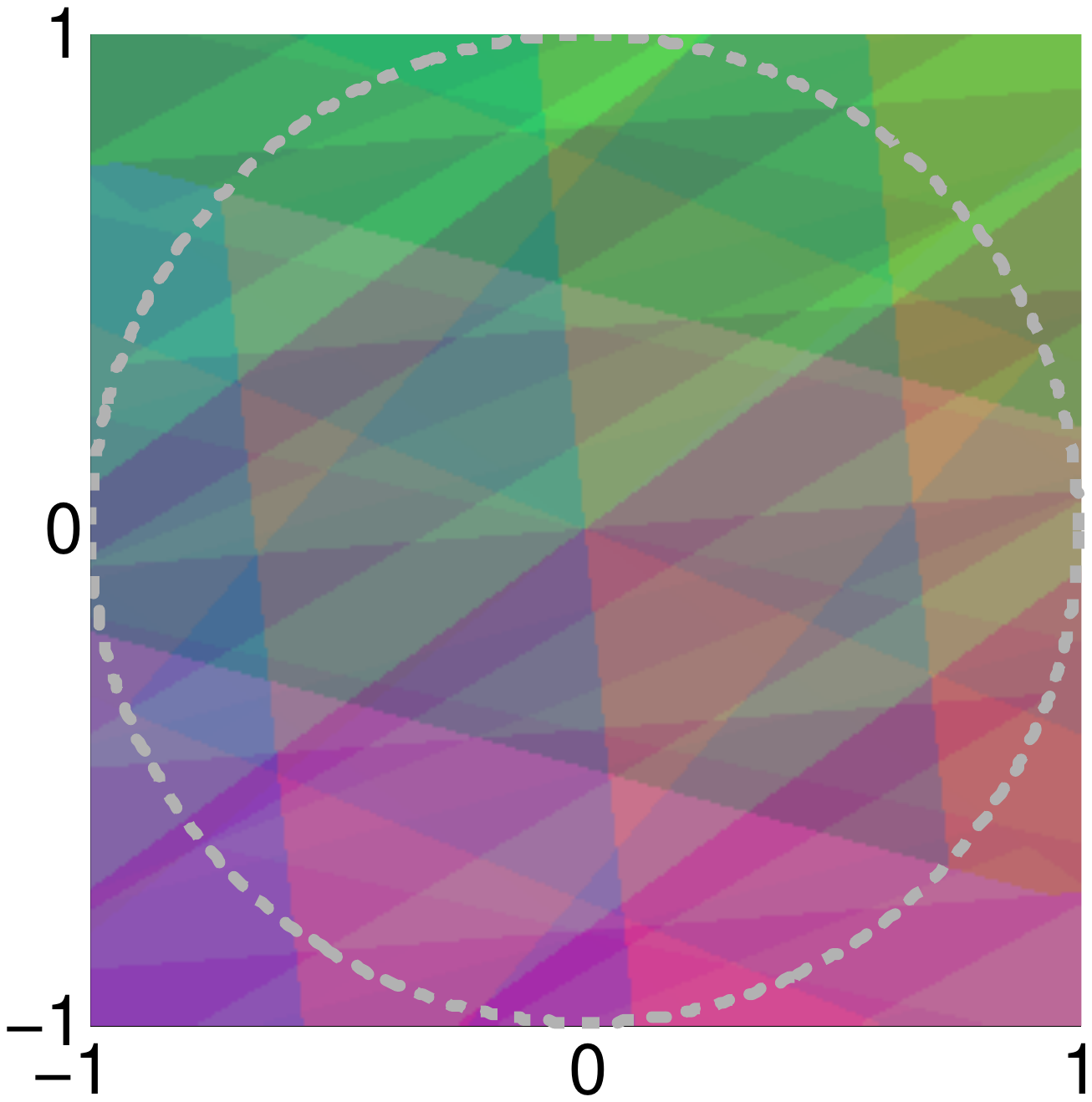}& 
  \includegraphics[width=.5\linewidth]{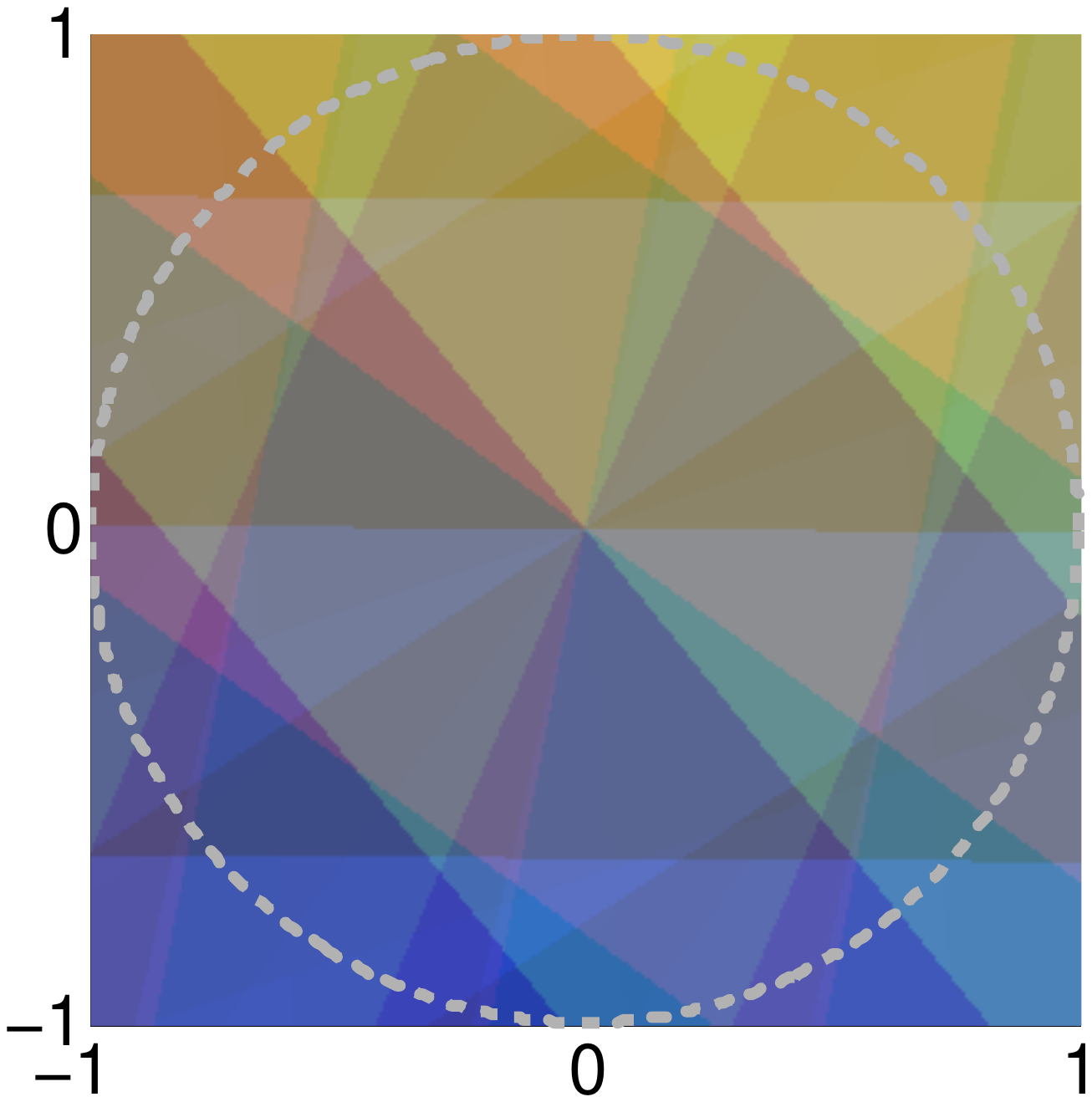}
  \end{tabular}
  \vspace{-40pt}
  \caption{The first order $\sd$ (left) and scalar quantization (right) cells associated with 2-bit quantization of $\Phi^T x$ where $x$ is in the unit ball of $\mathbb{R}^2$ and $\Phi$ is a $2\times 15$ Gaussian random matrix.}
  \label{fig:quant_cells}
  \end{figure}

The existing recovery algorithms for $\sd$ quantized compressed sensing measurements rely on a two stage algorithm. In the first stage, the signal support is recovered and in the second stage, the signal coefficients are estimated using the Sobolev dual of the frame associated with the recovered support. 

\subsection{Finding the Signal Support}\label{subsec:supp}
Let $q\in \qQal^m$ be the $r$-th order $\sd$ quantization of the compressed sensing measurements $y=Ax \in \mathbb{R}^m$ associated with the sparse vector $x\in\Sigma_k$ and the measurement matrix $A\in \mathbb{R}^{m\times n}$. In order to preserve the codebook definition \eqref{eq:SD_codebook}, we assume as in Sec.~\ref{sec:scalarCS} that the scaling of the entries of $A$ is independent of $m$.

The goal of the first stage of the reconstruction algorithm is to recover  $T:= supp(x)$. To that end, following \cite{GLPSY13} we will use a (standard) compressed sensing decoder $\decoder: \mathbb{R}^m \to \mathbb{R}^N$ that has uniform robustness guarantees for matrices with an appropriate RIP constant. For such a decoder and an arbitrary scalar $\kappa$
\begin{equation}\label{eq:CS_guarantee}
x\in\Sigma_k \text{ and } \gamma\in\mathbb R^m: \|\gamma\|_2 \leq \kappa\sqrt{m} \quad \implies \quad \|\decoder(A x + \gamma) - x \|_2 \leq C\kappa. 
\end{equation} 
For example, if $\decoder(Ax+\gamma)$ is the output of an $\ell_1$-minimization algorithm such as Basis Pursuit DeNoising (BPDN), it satisfies \eqref{eq:CS_guarantee} with constant $C:=C(\delta, k)$ when the matrix $A$ (more precisely ${A}/{\sqrt{m}}$) satisfies an appropriate restricted isometry property \cite{candes2008rip}. As the next proposition shows, robust decoders allow recovering the support of a sparse vector when its smallest non-zero entry is above the error level.

\begin{proposition}\label{prop:support_recovery}Let $\decoder$ be a compressed sensing decoder satisfying \eqref{eq:CS_guarantee} and let $x\in \Sigma_k$ with $T:=supp(x)$. Define $\hat{x}:=\decoder(A x + \gamma)$.  If $\min_{i\in T} |x_i| > 2C\kappa$ then the largest $k$ coefficients of $\hat{x}$ are supported on $T$. 
\end{proposition}
\begin{proof}First, note that
for all $i\in T$,  \eqref{eq:CS_guarantee} yields $|\hat{x}_i - x_i| \leq C\kappa$.  Since $\min_{i\in T} |x_i| > 2C\kappa$, the reverse triangle inequality gives $|\hat{x}_i| > C\kappa$ for all $i$ in $T$. On the other hand, \eqref{eq:CS_guarantee} also ensures that $|\hat{x}_i| \leq C\kappa$ for all $i \in T^c$.

 A sharper version of this argument appears in \cite{GLPSY13} but Proposition \ref{prop:support_recovery} is sufficient for our purposes. In particular, consider an $r$th order greedy $\sd$ quantization associated with  a codebook $\qQal$ having $2L$ elements.  Applying such a scheme to $Ax$ yields a quantized vector $q$ satisfying $\|q-Ax \|_2 \leq \frac{\Delta}{2} 2^r \sqrt{m}$ provided 
\begin{equation}
 L>\|Ax \|_\infty/\Delta + 2^{r-1} - 1/2 \label{eq:number_of_levels}.
\end{equation}
 Thus assuming that $A/\sqrt{m}$ has appropriate RIP constants,  Proposition \ref{prop:support_recovery} shows that using a decoder satisfying \eqref{eq:CS_guarantee}, the support $T$ of $x\in\Sigma_k\subset \mathbb{R}^n$ can be accurately recovered provided $|z_i| > 2^rC~ \Delta$ for all $i \in T$. What remains is to choose the number of levels $L$ in the codebook to satisfy \eqref{eq:number_of_levels}; this in turn requires an estimate of $\|Ax\|_\infty$. 
 
 To that end, we now consider subgaussian measurement matrices, \ie matrices whose entries are subgaussian random variables as defined below. 
 
 \begin{definition}
Let $\xi$ be a Gaussian random variable drawn according to $\mathcal{N}(0,\sigma^2)$. If a random variable $\eta$ satisfies $P(|\eta|>t) \leq e P(|\xi|>t)$ for all $t$, then we say $\eta$ is subgaussian with parameter $\sigma>0$.
\end{definition}

Examples of subgaussian random variables include Gaussian, Bernoulli, and bounded random variables, as well as their linear combinations. For matrices populated with such subgaussian entries, the following proposition from \cite{KSY13} gives a bound on $\|Ax\|_\infty$ when the non-zero entries of $x$ are restricted to a fixed support $T$ so that $Ax = \Phi^Tx_T$ for a frame $\Phi$ associated with the support.

 \begin{proposition}\label{prop:infty_norm_bound}
 Let $\widehat{\Phi}$ be a $k \times m$ subgaussian matrix with mean zero, unit variance, and parameter $\sigma$, where $k<m$. Let $\Phi= \frac{1}{\sqrt{m}}\widehat{\Phi}$ and fix $\alpha\in (0,1)$. Then, with probability at least $1-e^{-\frac{1}{4} m^{1-\alpha}k^\alpha}$, we have  for all $m> C^{\frac{1}{1-\alpha}} k$ and $x \in \mathbb{R}^k$
\begin{equation}\|\Phi^Tx\|_{\infty} \leq  e^{1/2}\big(\tfrac{m}{k}\big)^{-\frac{\alpha}{2}}\|x\|_2.\label{eq:infty_norm_bound}\end{equation}
Here $C$ is a constant that may depend on $\sigma$, but is independent of $k$ and $\alpha$.
\end{proposition}
Taking a union bound over all the ${n\choose k}$ submatrices of $A$ of size $m\times k$ yields an identical uniform bound on $\|Ax\|_\infty$, which holds for sparse vectors $x$ with high probability, provided $m>C k (\log n)^{\frac{1}{1-\alpha}}$.  

Thus an $r$-th order greedy $\sd$ scheme with sufficiently many quantization levels allows the recovery of a sparse signal's support from its compressed sensing measurements. Equipped with this knowledge, we can estimate the signal coefficients using the Sobolev dual of the frame associated with the recovered support. 

\end{proof}

\subsection{Recovering the Signal Coefficients}\label{subsec:coeff}

\newcommand{\sing}{S}

We continue to consider Gaussian or subgaussian measurement matrices, now assuming that the support $T$ of the signal $x$ has been identified. Our goal is to approximate the coefficients $x_i$, $i\in T$. With high probability, the matrix $A/\sqrt{m}$ has the restricted isometry property of order $2k$ and level $\delta_{2k}\leq 1/\sqrt{2}$ provided one takes at least on the order of $k\log(n/k)$ measurements. Then the matrix $A_T/\sqrt{m}$ restricted to the columns indexed by $T$ is close to an isometry and its rows hence form a frame. Consequently, the measurement vector is the associated frame expansion of $x_T$, and $q$ is the corresponding $\sd$ frame quantization.

As shown in Sec.~\ref{subsec:frame}, it is advantageous to reconstruct $x$ from the $r$-th order $\sd$ quantization $q$ of the measurement vector $Ax$ using the Sobolev dual $\widetilde A_T^{(r)}$ of $A_T$, see \eqref{eq:errbound} and \eqref{eq:Sobolev}.  A possible bound for the reconstruction error is then proportional to $\|\widetilde A_T^{(r)}D^r\|_{2\rightarrow 2}$. Thus to show a uniform recovery guarantee, one needs a bound for this quantity which is uniform over all potential support sets $T$. In the initial work \cite{GLPSY13}, dealing with Gaussian compressed sensing matrices, the approach to proving such a bound consisted of explicitly controlling the lowest singular value of $D^{-r}A_T$. Their approach utilized the unitary invariance of the Gaussian measure to identify the distribution of the singular values of the random matrix $D^{-r}A_T$ with those of $\sing_{D^{-r}}\Psi$, where $\sing_{D^{-r}}$ is a diagonal matrix whose entries are the singular values of $D^{-r}$, and $\Psi$ is a Gaussian matrix. This, 
coupled with bounds on the singular values of $D^{-r}$, 
allowed \cite{GLPSY13} to derive bounds that held with probability high enough to survive a union bound over all ${n}\choose{k}$ Gaussian submatrices of $A$. In \cite{KSY13}, this approach was extended to subgaussian matrices.  Herein, to prove such a bound on $\|\widetilde A_T^{(r)}D^r\|_{2\rightarrow 2}$, we follow the simpler, RIP-based approach presented in \cite{FK13}. 

To that end, let $E=U_{E} \sing_{E} V_{E}^T$ be the singular value decomposition (SVD) of any matrix $E$ (for some orthogonal matrices $U_{E}$ and $V_{E}$) where the matrix  $\sing_{E}$ is diagonal with (ordered) diagonal entries $\sigma_j(E)$. We denote also $\sigma_{\min}(E) := \sigma_1(E)$ the smallest singular value of $E$.
Then the following proposition (see, \eg \cite{GLPSY13}) holds. 
\begin{proposition}
There are positive constants $C_1(r)$ and $C_{2}(r)$, independent of $m$, such that
\begin{equation}\label{D^-r}
C_1(r)(\tfrac{m}{j})^r\leq \sigma_j(D^{-r})\leq C_{2}(r)(\tfrac{m}{j})^{r},\ j=1,\ldots,m.
\end{equation}
\end{proposition}

Denote by $P_\ell$ the $\ell\times m$ matrix that maps a vector to its
first $\ell$ components. Moreover, denote by
$\widetilde\Sigma_k(A,\decoder) \subset \Sigma_k$ the set of
$k$-sparse signals $x$ whose support can be recovered from $q$ with
the decoder $\decoder$ as in Proposition
\ref{prop:support_recovery}. The following theorem describes the
reconstruction performance.

 \begin{theorem}[\cite{FK13}]\label{thm:SD_RIP}
Let $A\in \mathbb{R}^{m\times n}$ be a matrix such that for a fixed $\ell\leq m$, the $\ell \times n$ matrix $\frac{1}{\sqrt\ell}P_\ell V^T_{D^{-r}}A$ has restricted isometry constant $\delta_k \leq \delta$. Then the following holds uniformly for all $x\in \widetilde\Sigma_k(A,\decoder)$. 

If $x$ has support $T$, $q$ is the $r$-th order $\sd$ quantization of $Ax$, and $\hat{x} := \widetilde{A}_{T}^{(r)}q$, then  
$$
\|x-\hat{x}\|_2 \leq \tfrac{\qD}{C(r)\sqrt{(1-\delta)}}(\tfrac{m}{\ell})^{-r+\frac{1}{2}},
$$
where $C(r)>0$ is a constant depending only on r and $\qD$ is the quantization step size.
\end{theorem}
 \begin{proof}
As the SVD of $D^{-r}$ provides $D^{-r}=U_{D^{-r}}\sing_{D^{-r}}V_{D^{-r}}^T$, the smallest singular value of $D^{-r}A_T$ satisfies
 \begin{align*}
 \sigma_{\min}(D^{-r}{A_T})&=\sigma_{\min}(\sing_{D^{-r}}V^T_{D^{-r}}{A_T})\\
 &\geq \sigma_{\min}(P_\ell \sing_{D^{-r}}V^T_{D^{-r}}{A_T})\\
 &= \sigma_{\min}((P_{\ell}\sing_{D^{-r}}P^T_{\ell} )(P_\ell V^T_{D^{-r}}{A_T}))\\
 &\geq \sigma_\ell(D^{-r}) \sigma_{\min}(P_\ell V^T_{D^{-r}}{A_T}),
 \end{align*}
To bound $\sigma_{\min}(P_\ell V^T_{D^{-r}}{A_T})$ uniformly over all support sets $T$ of size $k$  we simply note that if $\frac{1}{\sqrt{\ell}}P_\ell V^T_{D^{-r}}{\Phi}$ has restricted isometry constant $\delta_k\leq\delta$
 then $\sigma_{\min}(P_\ell V^T_{D^{-r}}{A_T})$ is uniformly bounded from below by 
 \begin{equation}\label{subRIP}
 \sqrt{\ell}\sqrt{1-\delta}.
 \end{equation}
 The theorem follows by applying (\ref{eq:errbound}), (\ref{D^-r}), (\ref{subRIP}) as
\begin{equation}
 \tfrac{1}{\sigma_{min}(D^{-r}{A_T})}\|u\|_2\leq  \tfrac{\Delta}{C(r)\sqrt{(1-\delta)}}(\tfrac{m}{\ell})^{-r+\frac{1}{2}}
\end{equation}\qed
\end{proof}

The above theorem can be applied almost directly to Gaussian compressed sensing matrices. If $A$ is a Gaussian matrix with independent zero mean and unit variance entries, then by rotation invariance so is the matrix $P_\ell V^T_{D^{-r}}A$. Regarding the choice of $\ell$, note from Theorem \ref{thm:SD_RIP} that the smaller $\ell$ is, the better the bound. On the other hand $\ell$ has to be large enough for $\frac{1}{\sqrt{\ell}}(P_\ell V^T_{D^{-r}}\Phi)$ to have the restricted isometry constant $\delta_k\leq\delta$. This prompts the choice $\ell \asymp k\log n$, as then $\frac{1}{\sqrt{\ell}}(P_\ell V^T_{D^{-r}}\Phi)$ has the restricted isometry constant $\delta_k<\delta$ with high probability, as discussed in Chapter 1 (see, \eg Theorem 1.5). In particular, if \[m\ \gtrsim\ k(\log n)^\frac{1}{1-\alpha}, \quad \alpha \in (0,1)\]
and \[ \ell \asymp k\log n\]
 then \[\tfrac{m}{\ell}\asymp \tfrac{m}{k\log n} = (\tfrac{m}{k})^\alpha\cdot \bigg(\tfrac{m}{k(\log n)^{\frac{1}{1-\alpha}}}\bigg)^{1-\alpha}\gtrsim (\tfrac{m}{k})^{\alpha}\]
Applying Theorem \ref{thm:SD_RIP} directly, we obtain
  \begin{align*}
\|x-\hat{x}\|_2&\lesssim\Delta(\tfrac{m}{k})^{-\alpha (r-\frac{1}{2})}.
\end{align*}
This essentially recovers the result in \cite{GLPSY13} and a similar, albeit more technical argument for subgaussian matrices, using either bounds on tail probabilities for quadratic forms \cite{HW71, rv13} or bounds for suprema of chaos processes \cite{krmera12} recovers the analogous result in \cite{KSY13}. 

To illustrate the advantage of using $\sd$ schemes for quantizing compressed sensing measurements we conduct a numerical experiment with $k$-sparse signals in $\mathbb{R}^n$, as we vary the number of measurements $m$. We fix $k=10,$ $n=1000,$ and the quantization step-size $\Delta=0.01$. We draw  $m \times n$ Gaussian matrices $A$ for $m\in\{100, 200, 400, 800\}$ and quantize the measurements $Ax$ using scalar quantization and $r$th order $\sd$ schemes with $r=1,2,3$. We then use the two-stage reconstruction method described herein to obtain an approximation $\hat{x}$ of $x$ using its quantized measurements. Repeating this experiment 30 times, we compute the average of the reconstruction error $\|x-\hat{x}\|_2$ for each of the quantization methods and plot them against the oversampling ratio $m/k$ in Fig.~\ref{fig:SD_CS_results}.

\begin{figure}[t]
  \centering
  \includegraphics[width=0.9\linewidth]{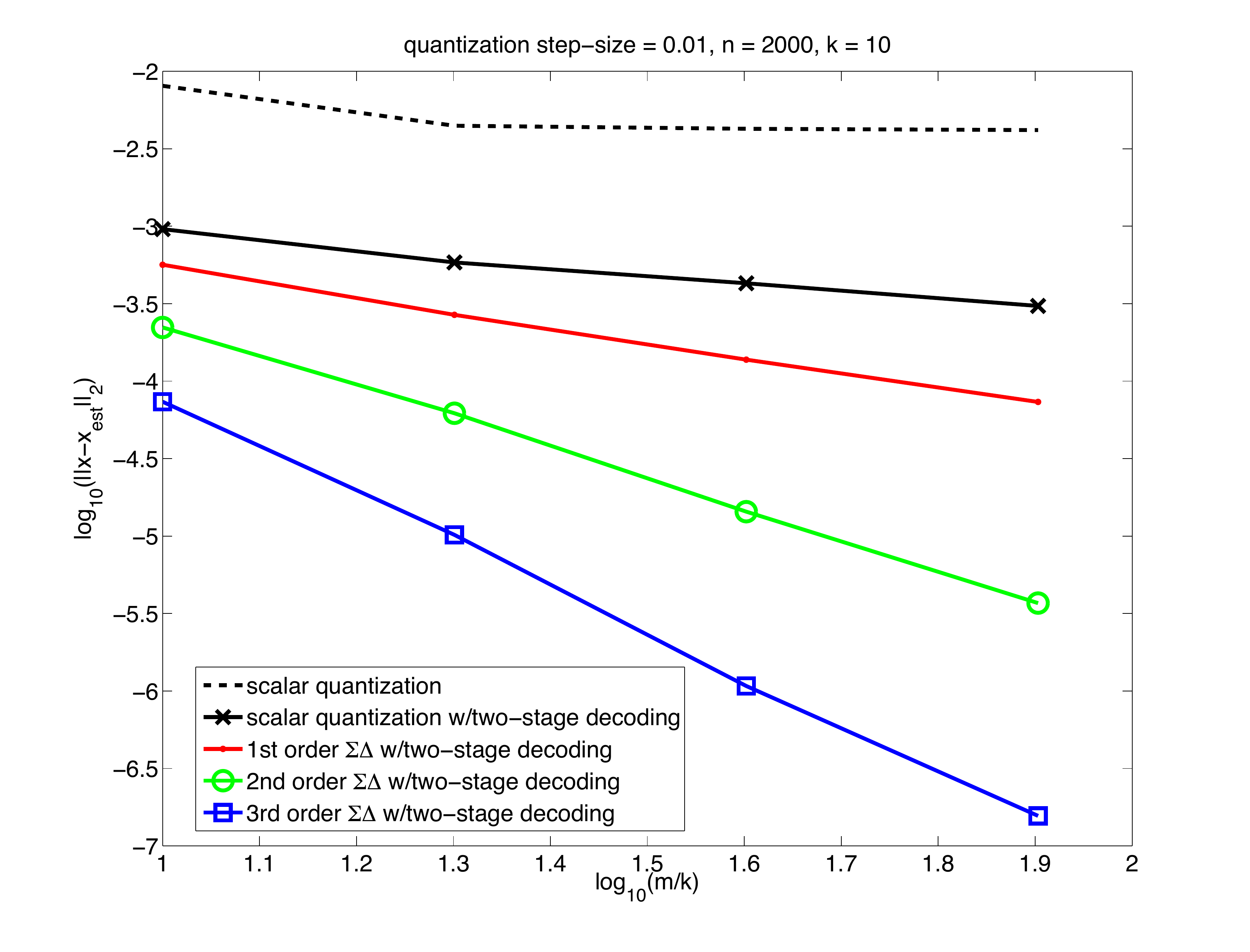} 
  \vspace{-10pt}
  \caption{Average errors over 30 experiments. The figure shows the reconstruction errors resulting from scalar quantization, using $\ell_1$-minimization for reconstruction (dashed line). Also corresponding to scalar quantization, the figure shows the errors resulting from reconstructing via the two-stage algorithm described herein (solid black line), using the canonical dual of the frame corresponding to the recovered support in the second stage. It also shows the reconstruction errors resulting from $1st$, $2nd$, and $3rd$ order $\sd$ quantization respectively. These errors decay as $(m/k)^{-r}$ for $r=1,2,3$ respectively, slightly outperforming the theoretical predictions presented here.}
  \label{fig:SD_CS_results}
  \end{figure}

In summary, using Gaussian and subgaussian compressed sensing matrices
recovery of sparse signals from their $\sd$ quantized measurements is
possible. More importantly, the reconstruction error decays
polynomially in the number of measurements and thus outperforms the
(at best) linear error decay that can be achieved with scalar
quantization. This improvement comes at the cost of introducing memory
elements, and feedback, into the quantization procedure.
\section{Discussion and Conclusion}
\label{sec:discussion}
Quantization is an essential component of any acquisition system, and,
therefore, an important part of compressive sensing theory and
practice. While significant work has been done in understanding the
interaction of quantization and compressive sensing, there are several
open problems and questions.

One of the most interesting open problems is the interaction of
quantization with noise. While the discussion and references in
Sec.~\ref{sec:extensions} provides some initial results and theoretical analysis, a
comprehensive understanding is still missing. An understanding of the
optimal bit allocation and the optimal quantizer design, uniform or
non-uniform scalar, or \sd, given the noise level, as well as the
robustness of the reconstruction to noise and quantization is still
elusive. 

While \sd\ can be used to improve the rate efficiency of compressive
sensing, compared to scalar quantization, the performance is still not
comparable to the state-of-the-art in conventional \sd\ methods. For
example, conventional \sd\ quantization of band-limited functions can achieve error that decays exponentially
as the sampling rate increases, not currently possible with existing
compressive sensing \sd. Furthermore, the analysis in
Sec.~\ref{sec:SDCS} does not hold for 1-bit quantization, often
desirable in practical systems due to its simplicity. Such an
extension has significant practical importance.

Even with \sd\ approaches, the rate efficiency of compressive sensing
systems is not ideal. As evident from the fundamental bounds in
Sec.~\ref{sec:intro_fundamentals}, compressive sensing is not
rate-efficient compared to classical methods such as transform
coding. In others word, while compressive sensing is very promising
in building sensing systems because it can significantly reduce the
number of measurements and the sampling burden, it is not a good data
compression approach if the measurements have already been obtained
and the achievable bit-rate is important. That said, due to the intimate connection between frame quantization and quantization for compressed sensing, promising results in the finite frames context, e.g., \cite{iwen2013near} can inform future developments in compressed sensing. 

The potential encoding simplicity of a compressive sensing
system is very appealing. Acquiring generalized linear measurements and quantizing
them can be less complex than typical transform-coding
approaches and much more attractive in low-power and
computationally-restricted sensing applications. The complexity is
shifted to the reconstruction, which, in many applications, can bear
significantly more computational complexity. Nevertheless, the rate inefficiency of
compressive sensing can be a barrier in such applications.

A number of promising approaches have been proposed to overcome this
barrier using modifications of the quantizer that produce
non-contiguous quantization
regions~\cite{bib:Pai06,bib:BoufSAMPTA2011,B_TIT_12,Kamilov12}. Initial
theoretical analysis and experimental results are promising. However,
our understanding is still limited. One of the drawbacks of such
approaches is that the reconstruction is no longer convex and,
therefore, not as simple to provide guarantees for.

Alternatively, recent work on adaptive quantization strategies has
shown that error decay exponential in the bit-rate can be achieved,
even using a 1-bit quantizer, at the cost of adaptivity in the
measurements and -- in contrast with the methods presented in this chapter -- significant computation at the encoder. Specifically,~\cite{baraniuk2014exponential} shows that
adaptively choosing the threshold of a 1-bit quantizer allows
the error to decay exponentially with the number of measurements. The cost is that the thresholds are updated by solving an $\ell_1$ minimization problem, or running an iterative hard thresholding scheme. It is thus interesting to quantify the tradeoff between computational complexity at the quantizer, and achievable reconstruction accuracy.

Another important aspect is that while the best recovery guarantees in
compressed sensing are obtained for Gaussian and subgaussian
measurement matrices, which are also mainly considered in this
article, applications usually require structured matrices, such as
subsampled Fourier matrices, e.g., as a model for subsampled MRI
measurements \cite{ldp07}, or subsampled convolution, e.g., as a model
for coded aperture imaging \cite{mawi08}. In both cases, when the
subsampling is randomized, near-optimal recovery guarantees are known
for unquantized compressed sensing \cite{RV08:sparse,krmera12}.
Combined with quantization, however, hardly anything is known for such
matrices. Such results would be of great importance to move the
approaches discussed in this survey closer to the application
scenarios.

Quantization is also important when considering randomized embeddings,
an area of research intimately related to compressive
sensing~\cite{baraniuk2008simple,krahmer2011new}. Embeddings are
transformations that preserve the geometry of the space they operate
on; reconstruction of the embedded signal is not necessarily the
goal. They have been proven quite useful, for example, in signal-based
retrieval applications, such as augmented reality, biometric
authentication and visual
search~\cite{LRB_MMSP12,BR_DCC13,SBV_SPIE13_Embeddings}.

These applications require storage or transmission of the embedded
signals, and, therefore, quantizer design is very important in
controlling the rate used by the embedding. Indeed, significant
analysis has been performed for embeddings followed by conventional
scalar quantization, some of it in the context of quantized
compressive
sensing~\cite{jacques2013robust,plan2012robust,plan2011dimension} or
in the study of quantized extensions to the Johnson Lindenstrauss Lemma~\cite{johnson1984extensions,LRB_MMSP12,SBV_SPIE13_Embeddings,jacques2013quantized}. Furthermore,
since reconstruction is not an objective anymore, non-contiguous
quantization is more suitable, leading to very interesting quantized
embedding designs and significant rate reduction~\cite{BR_DCC13}. In
this context, quantization can also provide to significant computation
savings in the retrieval, leading to Locality Sensitive Hashing (LSH)
and similar methods~\cite{Andoni08LSH}.

\section*{Acknowledgement}
Petros T. Boufounos is exclusively supported by
Mitsubishi Electric Research Laboratories. Laurent Jacques is a Research Associate funded by the Belgian
F.R.S.-FNRS. Felix Krahmer acknowledges support by the German Science Foundation (DFG) in the context of the Emmy-Noether Junior Research Group KR 4512/1-1 ``RaSenQuaSI''.
Rayan Saab is an assistant professor of mathematics with the University of California, San Diego.

\end{document}